\documentclass[reqno]{amsart}
\input{commands}

\title[UAT for noisy QNN]{Quantitative Universal Approximation for Noisy Quantum Neural Networks}
\author{Lukas Gonon}
\address{School of Computer Science, University of St. Gallen and Department of Mathematics, Imperial College London}
\email{l.gonon@imperial.ac.uk}
\author{Antoine Jacquier}
\address{Department of Mathematics, Imperial College London}
\email{a.jacquier@imperial.ac.uk}
\author{Marcel Mordarski}
\address{Department of Computing, Imperial College London}
\email{marcel.mordarski25@imperial.ac.uk}

\thanks{All authors contributed equally. We used generative AI to polish the manuscript, proofread, and assist with code development, but not to formulate theoretical models, conceptual frameworks, hypotheses, or mathematical claims. We manually reviewed all AI-assisted work and take full responsibility for the final content. AJ acknowledges financial support from the EPSRC grant EP/T032146/1. MM is funded through the UKRI DLA scholarship.}
\date{\today}

\begin{document}

\begin{abstract}
We prove a quantitative universal approximation theorem for noisy quantum neural networks (QNNs), with explicit non-asymptotic error bounds in terms of the number of accuracy blocks, the qubit count, the number of two-qubit gates, and hardware noise parameters. Motivated by quantitative finance, where target functions are expectations of payoffs, we first extend existing noiseless bounds to expectation functions and derive explicit constants for exponential L\'evy models, covering the Black–Scholes price of (truncated) European Put options. We then formulate the noisy circuit via quantum channels, obtain a fidelity-based bound for general noise models, and specialise to depolarising noise with readout error, calibrated to real hardware. Finally, we provide a detailed numerical analysis for option pricing in the Black–Scholes model, testing our results on actual noisy quantum hardware.
\end{abstract}
\vspace{-2cm}\maketitle

\setcounter{tocdepth}{2}
\tableofcontents

\section{Introduction}

Quantum computing promises speed-up for many computational problems, yet near-term hardware is fundamentally limited by noise. A central question is therefore not \emph{whether} a given function can be represented by an idealised quantum circuit, but \emph{how accurately} it can be approximated by a noisy circuit on hardware available today, and how that accuracy degrades with the noise parameters of the device. In this paper, we provide a precise, non-asymptotic answer to this question for a class of functions of direct applied interest -- expectations arising in quantitative finance -- and validate the resulting error bounds on a real IBM quantum processor.

Parameterised quantum circuits have emerged as a powerful paradigm for approximating high-dimensional functions. Following classical universal approximation theorems for neural networks~\cite{Barron1993,cybenko1989approximation,hornik1989multilayer,hornik1991approximation,kolmogorov1957representations,leshno1993multilayer,yarotsky2017error}, qualitative quantum analogues have appeared in~\cite{PerezSalinas2020datareuploading,perez2021one,Schuld2021}, and \emph{quantitative} versions with explicit, non-asymptotic error bounds (depending on circuit width, circuit depth, and number of qubits) have been derived for a range of function classes: Korobov~\cite{aftab2024korobov}, Fourier-integrable~\cite{gonon2023universal,gonon2026feedbackdriven}, Sobolev~\cite{manzano2025approximation}, H\"older~\cite{yu2024non}, and periodic ($C^k$-smooth)~\cite{neufeld2025qnn}. All of these results assume noiseless quantum circuits. This is a strong assumption: in the current \emph{noisy intermediate-scale quantum} (NISQ) regime, even modest noise radically alters circuit behaviour. Under non-unital noise, deep parameterised circuits estimating expectation values collapse to logarithmic depth~\cite{mele2024noise}, and in this truncated regime, their outputs can be classically estimated to constant precision~\cite{franca2025efficient}. The existing approximation theorems are therefore in tension with the emerging theory of noisy circuits: they guarantee expressive power in a regime that may be physically unattainable on near-term devices, raising the natural question:

\begin{center}
\textit{Which functions remain representable, with what accuracy, once noise is introduced in the system?}
\end{center}

On current hardware, uncorrected noise and decoherence set stringent depth budgets, barren plateaus make deep architectures hard to train, and naive transcriptions of expressive noiseless ans\"atze suffer dramatic performance losses~\cite{wang2021noise,larocca2025barren}: the accuracy gap between noise-free simulation and execution on IBMQ devices can exceed $60\%$ on simple classification tasks unless the model is explicitly made noise-aware~\cite{wang2021roqnn}. A growing toolbox of noise injection, post-measurement normalisation, and quantisation techniques has emerged to patch this gap~\cite{wang2022quantumnat}, and more recent noise-aware QNN architectures such as NQNN~\cite{rahman2025nqnn} maintain accuracy by embedding noise attenuation in the variational ansatz. These approaches are, however, empirical, and lack the theoretical foundations needed to \emph{quantify} the residual risk due to noise.

A second source of motivation comes from applied probability, in particular quantitative finance, where target functions are typically expectations of random processes -- option prices, risk measures, or sensitivities -- and both speed and accuracy are of primary importance. Classical tools rely on Monte Carlo techniques~\cite{glasserman2003monte} or PDE methods~\cite{morton1994numerical} and are by now very well understood. Quantum algorithms have recently entered this landscape: amplitude-estimation methods~\cite{stamatopoulos2020option} and quantum Monte Carlo~\cite{cui2024quantum,chen2026qmc} grant, asymptotically, a quadratic speed-up over the classical $\mathcal{O}(\eps^{-2})$ Monte Carlo rate, but under assumptions (large circuit depth, noiseless gates) incompatible with NISQ. A few NISQ-friendly proposals -- unary encodings~\cite{ramos2021quantum} or non-Hermitian variational formulations~\cite{kumar2024simulating} -- offer empirical assessments only, and rarely account for explicit hardware noise.

\paragraph{Contributions.}
Against this backdrop, we make four contributions.
\begin{enumerate}
\item[(i)] \emph{Universal approximation for expectations} (Section~\ref{sec:theory}). 
We extend the noiseless UAT of~\cite{gonon2023universal} from $L^1$-Fourier-integrable functions to functions of the form $f(\xx) = \EE[\Phi(\xx + L)]$ for an $\RR^d$-valued random variable $L$, with explicit bounds. 
We further provide explicit constants for exponential L\'evy models and for the Black-Scholes price of European options.
\item[(ii)] \emph{Universal approximation under arbitrary CPTP noise} (Section~\ref{sec:NoiseModel}, Theorem~\ref{thm:general_noise_bound}). We recast the QNN of~\cite{gonon2023universal} in the density-operator formalism and prove a fidelity-based bound that holds for \emph{any} CPTP noise channel, isolating the contribution of noise. 
\item[(iii)] \emph{Explicit bound under depolarising noise with readout error} (Theorem~\ref{thm:depolar_approximation}, Proposition~\ref{prop:full_noise_approximation} and Corollary~\ref{cor:readout_exact}). For depolarising noise, we decompose the total error into an approximation, a systematic, and an offset component, each expressed in terms of hardware parameters, and show that readout error further contracts the output.
We further show (Theorem~\ref{cor:noise_correction}) that a two-parameter post-processing layer \emph{exactly} cancels the depolarising bias, and (Corollary~\ref{cor:noise_aware}) that parameters trained against the noisy output remove the systematic term altogether -- two practical recipes for noise-aware training.
\item[(iv)] \emph{Hardware experiment} (Section~\ref{sec:experiments}). We train Black-Scholes Put pricing QNNs and execute them on the IBM \texttt{ibm\_fez} Heron r2 processor. The hardware output is a contraction of the noiseless output, as predicted by the depolarising model, but the contraction is markedly stronger than calibration data predict ($0.49\pm0.07$ vs $0.86$). The root-mean-square deviation from the Black-Scholes price ($0.048$) lies inside our analytical envelope $\eps_{\mathrm{total}}=0.264$, and fitting the post-processing layer of~(iii) to data halves the error.

\end{enumerate}

\noindent\textbf{Scope and limitations.}
Our bounds decay at the dimension-independent rate $n^{-1/2}$, while their constants depend on the dimension~$d$ of the input. 
We restrict attention to the QNN architecture of~\cite{gonon2023universal}. The techniques transfer to other Fourier-type ans\"atze, but the depolarising bound is architecture-specific. Finally, hardware execution is illustrated on a single superconducting processor and a one-dimensional pricing problem with $n=8$ accuracy blocks (five qubits), for which the bound exceeds the error of the trivial predictor~$0$; a simulation at $n=32$ illustrates the informative regime (Section~\ref{subsec:noise_expt}). The depolarising model is an idealisation: on \texttt{ibm\_fez}, the fidelity factor computed from calibration data overestimates the measured one (Section~\ref{subsec:hardware_expt}). We discuss limitations in Section~\ref{sec:conclusion}.\\

\noindent\textbf{Outline.} 
Section~\ref{sec:theory} extends the noiseless quantitative UAT of~\cite{gonon2023universal} to expectation functions and gives explicit constants for exponential L\'evy models, with Black-Scholes as a worked example. Section~\ref{sec:NoiseModel} develops the noisy-circuit formalism via CPTP channels, derives the general fidelity bound, and specialises to depolarising noise calibrated to real hardware. Section~\ref{sec:experiments} reports numerical experiments on the Black-Scholes Put, both in simulation and on \texttt{ibm\_fez}. Appendices~\ref{app:table} and~\ref{app:figures} contain the hardware-parameter table for three contemporary devices (IBM, Quantinuum, Rigetti) and all figures supporting the experiments of Section~\ref{sec:experiments}.

\paragraph{Main result at a glance.}
For any $R>0$, $f \in \Ff_R$, $n \in \NN$, and depolarising noise with readout error, there exist parameters~$\theta$ such that the noisy QNN output $\overline{f}^{R}_{n,\theta}$ satisfies
\begin{equation*}
    \begin{split}
    \left(\int_{\RR^{d}} |f(\xx)-\overline{f}^{R}_{n,\theta}(\xx)|^{2}\mu(\D\xx)\right)^{\half} \leqslant{}& \underbrace{\tfrac{\alpha L^1[\widehat{f}]}{\sqrt{n}}}_{\text{approximation}} + \underbrace{(1-\alpha)\|f\|_{L^{2}(\mu)}}_{\text{systematic}}\\
    &+ \underbrace{R(1-\alpha)\bigl(1-\tfrac{4n}{2^{\nf}}\bigr)}_{\text{offset}} + \underbrace{4R\pf}_{\text{readout}},
    \end{split}
\end{equation*}
where $\alpha=(1-\lambda_{\Vg})(1-\lambda_{\Ug}) \in [0,1]$ is the hardware fidelity factor. Section~\ref{subsec:hardware_params} gives explicit formulae for~$\alpha$ in terms of hardware parameters~$\boldsymbol{\eps}$ read from device calibration data, which also provide~$\pf$. The bound is evaluated on the IBM \texttt{ibm\_fez} Heron r2 processor in Section~\ref{subsec:hardware_expt}, where the output has the affine form predicted by the depolarising model, but with a stronger contraction than calibration data predict. Corollary~\ref{cor:readout_exact} sharpens the bound to $\frac{\bar\alpha L^1[\widehat f]}{\sqrt n}+(1-\bar\alpha)\|f\|_{L^2(\mu)}+R(1-\alpha)(1-\frac{4n}{2^{\nf}})$ with $\bar\alpha=(1-2\pf)^2\alpha$, and, when $n_0=0$, Corollary~\ref{cor:noise_aware} shows that parameters trained against the noisy output attain the noiseless bound $L^1[\widehat f]/\sqrt n$ whenever $\bar\alpha R\geqslant L^1[\widehat f]$.

\section{Quantum approximations of expectation functions}\label{sec:theory}

We prove a quantitative  Universal Approximation Theorem for expectation functions.  We consider QNNs as in~\cite{gonon2023universal}, which were also implemented in~\cite{agarwal2024extending} on a Rydberg atom array.

\subsection{Notations and background}

We recall the following spaces introduced in~\cite{gonon2023universal}:
\begin{equation}\label{eq:Fspaces}
\begin{array}{rl}
\Ff & := \Big\{f:\RR^d\to\RR: f\in \Cc\left(\RR^d\right)\cap L^{1}\left(\RR^d\right)\Big\},\\
\Ff_{R} & := \left\{f\in\Ff,
\text{ with }\LOnefhat \leqslant R\right\},
\qquad\text{for any }R>0.
\end{array}
\end{equation}
The Fourier transform reads $\widehat{f}(\bx) := \int_{\RR^d} \E^{-2 \I\pi \xx \cdot \bx} f(\xx) \D \xx$, $\bx \in \RR^d$
and
$\LOnefhat:=\int_{\RR^d}| \widehat{f}
(\bx) | \D\bx$ 
(finite for $f \in \Ff$). 
Gonon and Jacquier~\cite{gonon2023universal} showed that a function $f\in\Ff_{R}$ may be approximated by (the output of) a quantum neural network (QNN) up to accuracy at least~$\eps$ using  $\Oo (\lceil \log_2(\eps^{-1}) \rceil)$ qubits.
More specifically, define the function 
\begin{equation}\label{eq:f_QNN_NoNoise}
\ftrn(\cdot) := R\left\{1-2\Big(\PP_1
+\PP_2\Big)\right\},
\end{equation}
with 
$\PP_1,\PP_2$
the output probability of the parameterised quantum circuit (or QNN) specified in~\cite[Section~II.B]{gonon2023universal}.
Here
$\ttheta\in\TTheta:=(\RR^d\times\RR\times[0,2\pi])^n$ denotes the parameters of the QNN and $n\in\NN$ the accuracy parameter, related to the number of qubits~$\nf$ via $\nf=\lceil\log_2(4n)\rceil$. As in~\cite{gonon2023universal}, we let $n_0\in\NN_0$ be the smallest integer such that $\log_2(4n+n_0)\in\NN$, so that $2^{\nf}=4n+n_0$; in particular, $n_0=0$ if and only if~$n$ is a power of two. The probabilities~$\PP_{1}$
and~$\PP_{2}$ depend on~$\ttheta$ and the input variable of~$\ftrn(\cdot)$, but we omit the dependence for notational conciseness.
We shall be considering the following statement, with~$\mu$ an arbitrary fixed probability measure on~$\RR^d$:

\begin{statement}\label{st:state} It holds that
$\LOnefhat<\infty$ and there exists $\Bf>0$ such that, for any
$n\in\NN$, $R\geqslant \LOnefhat$, there exists~$\ttheta\in\TTheta$ such that
\begin{equation}
\left(\int_{\RR^d} 
 \left|f(\xx) - \ftrn(\xx)\right|^2 \mu(\D \xx)\right)^{\half} \leqslant  \frac{\Bf}{\sqrt{n}}.
\end{equation}
\end{statement}

In particular,~\cite[Theorem~II.4]{gonon2023universal}  showed that this  holds for $f\in\Ff_{R}$ with $\Bf = \LOnefhat$.
The remainder of this section is devoted to extending this result to expectation functions.

\subsection{Extensions to expectation functions}\label{subsec:expectations}

We start with the following case, which leverages the results in~\cite{gonon2023universal}.

\begin{proposition}\label{prop:convolutional}
Let~$L$ be an $\RR^d$-valued random variable and  $f(\xx) := \EE[\Phi_{1}({\xx+L})]$ on~$\RR^d$. If $\Phi_{1} \in L^1(\RR^d)$ and $\bx \mapsto \EE[\E^{\I L \cdot \bx }]$ is integrable,
then Statement~\ref{st:state} holds with $\Bf = \LOnefhat$.
\end{proposition}
\begin{proof}
First, we aim to show that $f$ and $\widehat{f}$ are integrable. 
By change-of-variable, the characteristic function $\bx \mapsto \EE[\E^{-\I L \cdot \bx }]$ of $-L$ is integrable. Hence, by~\cite[Proposition~2.5(xii)]{ken1999levy} the random variable $-L$ has a bounded, continuous density $g$ with respect to the Lebesgue measure and we may write $f(\xx) = \int_{\RR^d} \Phi_{1}(\xx-\yy) g(\yy)\D\yy = (\Phi_{1} * g)(\xx)$. In particular, this implies that~$f$ is integrable, since both $g$ and $\Phi_{1}$ are integrable, and continuous, since $|f(\xx+\mathbf{h})-f(\xx)|\leqslant\|g\|_{\infty}\|\Phi_{1}(\cdot+\mathbf{h})-\Phi_{1}\|_{L^1}\to0$ as $\mathbf{h}\to0$; hence $f\in\Ff$. Next, we show that  $\widehat{f}$ is integrable. 
The convolution theorem yields $\widehat{f}(\bx) = \widehat{\Phi}_{1}(\bx)\widehat{g}(\bx)$. Inserting the definition of $g$, we may represent 
$\widehat{g}(\bx)= \int_{\RR^d} \E^{-2\I\pi \xx \cdot \bx} g(\xx)\D\xx = \EE[\E^{-2\I\pi\bx \cdot (-L)}]$.
Hence, $\widehat{f}(\bx) = \widehat{\Phi}_{1}(\bx)\widehat{g}(\bx) = \widehat{\Phi}_{1}(\bx)\EE[\E^{2\I\pi\bx \cdot L}]$ is integrable, since (i) $\widehat{\Phi}_{1}$ is bounded due to $\Phi_{1} \in L^1(\RR^d)$ and (ii) $\bx \mapsto \EE[\E^{\I L \cdot \bx }]$ is integrable. 
Since~$f$ and~$\widehat{f}$ are both integrable, ~\cite[Theorem~II.4]{gonon2023universal} hence implies that  for any
$n\in\NN$, there exists~$\ttheta\in\TTheta$ such that 
\begin{equation}
\left(\int_{\RR^d} 
 \left|f(\xx) - \ftrn(\xx)\right|^2 \mu(\D \xx)\right)^{\half} \leqslant  \frac{\LOnefhat}{\sqrt{n}}
\end{equation}
for any choice of $R$ such that $\LOnefhat \leqslant R$. 
\end{proof}

As a special case, consider an exponential L\'evy model, so that
$L=(L_t)_{t \geqslant 0}$ is a $d$-dimensional L\'evy process with characteristic triplet $(\Sigma,\gamma,\nu)$,
where the $d\times d$ symmetric non-negative definite matrix ~$\Sigma$ denotes the diffusion matrix, $\gamma\in\RR^d$ is the drift and~$\nu$ a measure on~$\RR^d$ satisfying
$\nu(\{0\})=0$ and $\int_{\RR^{d}}(1\wedge\|\xx\|^2)\,\nu(\D\xx)<\infty$.
We refer the reader to~\cite[Chapters~1-2]{ken1999levy} for details about such processes.
Then $f(\xx) = \EE[\Phi(\E^{\xx+L_T})]$ corresponds to the price of a European option with payoff~$\Phi$ and maturity~$T$ in an exponential L\'evy model with~$d$ underlying asset prices modelled by $S_t = S_0 \E^{L_t}$ for $S_0 = \E^{\xx}$,
where the drift~$\gamma$ is chosen such that~$\E^{L}$ is a martingale.
For $\nu=0$, $d=1$ and $\gamma=-\half\Sigma$, this corresponds precisely to the  Black-Scholes model.
The next corollary shows that Proposition~\ref{prop:convolutional} can be directly applied to such models, with $\Phi_{1} = \Phi \circ \exp$. 
\begin{corollary}
Let $L$ be a $d$-dimensional L\'evy process such that $\bx \mapsto \EE[\E^{\I L_T \cdot \bx }]$ is integrable. 
Let $f(\xx) = \EE[\Phi(\E^{\xx+L_T})]$, $\xx \in \RR^d$ and assume $\Phi \circ \exp \in L^1(\RR^d)$. 
Then Statement~\ref{st:state} holds with $\Bf = \LOnefhat$.
\end{corollary}

The next corollary provides sufficient conditions ensuring that the integrability conditions on the characteristic functions are satisfied. The integrability conditions on the payoff~$\Phi$ are satisfied, for example, if~$\Phi$ is a butterfly basket option or a capped Call or Put option. 
\begin{corollary}\label{cor:Levy}
Let $L$ be a $d$-dimensional L\'evy process with non-degenerate diffusion matrix: there exists $C>0$ such that $\frac{1}{2}\bx \cdot  \Sigma \bx \geqslant C \|\bx\|^2$ for all $\bx \in \RR^d$. 
 Let $f(\xx) := \EE[\Phi(\E^{\xx+L_T})]$, on~$\RR^d$ and assume  $\Phi \circ \exp \in L^1(\RR^d)$. 
Then Statement~\ref{st:state} holds with $\Bf = \|\Phi \circ \exp\|_{L^1(\RR^d)}   \left(\frac{1}{4 \pi CT}\right)^{d/2}$. 
\end{corollary}

\begin{proof}
Following the proof of~\cite[Theorem~13]{Gonon2021}, then
$ |\EE[\E^{\I L_T \cdot \bx }] |\leqslant \E^{-CT\|\bx \|^2}$. 
In particular, $\bx \mapsto \EE[\E^{\I L_T \cdot \bx }]$ is integrable. Moreover, with $\Phi_{1} = \Phi \circ \exp$ and using the representation  $\widehat{f}(\bx)  = \widehat{\Phi}_{1}(\bx)\EE[\E^{2 \I \pi \bx \cdot L_T}]$ from the proof of Proposition~\ref{prop:convolutional}, 
we can  bound $\LOnefhat$ as 
\begin{align*}
\LOnefhat =\int_{\RR^d}| \widehat{f}
(\bx) | \D\bx  \leqslant \|\Phi_{1} \|_{L^1} \int_{\RR^d}\left|\EE[\E^{2 \I \pi  L_T \cdot \bx }]\right| \D\bx 
& \leqslant  \|\Phi_{1}\|_{L^1} \int_{\RR^d}\E^{-4 \pi^2 CT\|\bx \|^2} \D\bx\\
 & = \|\Phi_{1}\|_{L^1} \left|\frac{1}{4\pi CT}\right|^{\frac{d}{2}}.
\end{align*}
\end{proof}

\subsection{Examples}
\subsubsection{Gaussian density}
\label{sec:GaussianDensity}
Consider the Gaussian density
$f_{\sigma}(x) = \frac{1}{\sigma\sqrt{2\pi}} \exp\left\{-\frac{x^2}{2\sigma^2}\right\}$
on~$\RR$
and $\sigma>0$.
Then
$\widehat{f}_{\sigma}(\xi) = \E^{-2\pi^2\sigma^2 \xi^2}$ for all $\xi \in \RR$
and 
$L^1[\widehat{f}_{\sigma}] = \frac{1}{\sigma\sqrt{2\pi}}$.

\subsubsection{Put option in the Bachelier model}
The same reasoning also applies to (arithmetic) L\'evy models $f(\xx) = \EE[\Phi_{1}({\xx+L_T})]$, $\xx \in \RR^d$ with $\Phi_{1} \in L^1(\RR^d)$. For example, consider the Bachelier model $S_{t} = S_{0} + \sigma W_{t}$ with $\sigma>0$, for some standard Brownian motion~$W$. 
In this case, from Corollary~\ref{cor:Levy}
with $C=\frac{1}{2}\sigma^2$,
then Statement~\ref{st:state} holds with
$\Bf = \|\Phi_{1}  \|_{L^1(\RR)}   \frac{1}{\sqrt{2\pi T}\sigma}$.
For example, if $\Phi_{1}(x) = (K-x)^+ \boldsymbol{1}_{\{x\geqslant 0\}}$ is the payoff of a Put option (defined on the whole real line), then 
$\|\Phi_{1}\|_{L^1(\RR)} = \int_{0}^K (K-x) \D x = \frac{K^2}{2}$ and hence 
$\Bf =  \frac{K^2}{2\sqrt{2\pi T}\sigma}$. 
We can, however, get a sharper bound by directly computing the Fourier transform $\widehat{\Phi}_{1}(\xi) = \frac{K}{2 \I \pi \xi } + \frac{1-\E^{-2 \I \pi K \xi}}{(2\pi \xi)^2 }$ by integration by parts. Since $|\E^{\I u} -1 |\leqslant |u|$, then
$
|\widehat{\Phi}_{1}(\xi)| \leqslant  \frac{K}{\pi |\xi| }$ for $|\xi|>0$. Moreover, 
$ |\E^{-\I u} -1+\I u| \leqslant \frac{1}{2}|u|^2$ for all $u \in \RR$ and thus 
$|\widehat{\Phi}_{1}(\xi)| = | \frac{-1+\E^{-2 \I \pi K \xi}+2\I\pi \xi K }{(2\pi \xi)^2 }| \leqslant \frac{K^2}{2}$. 
Splitting the integral and inserting these bounds, we obtain, for any $a>0$,
\begin{align*}
\LOnefhat  & = \int_{\RR}   |\widehat{\Phi}_{1}(\xi)| |\EE[\E^{2 \pi \I L_T \cdot \xi }] | \D\xi\\
& \leqslant \int_{-1/a}^{1/a} \frac{K^2}{2} \E^{-\frac{1}{2}\sigma^2 T \xi^2 (2 \pi)^2}  \D \xi + 2
\int_{1/a}^\infty  \frac{K}{\pi |\xi| } \E^{-\frac{1}{2}\sigma^2 T \xi^2 (2 \pi)^2}  \D\xi\\
 &=  K^2\frac{1}{2\sigma\sqrt{2T\pi}}\mathrm{erf}\left(\frac{2 \pi \sigma\sqrt{T}}{a\sqrt{2}}\right) 
 +  \frac{K}{\pi}   \mathrm{E}_1\left(
\frac{2 \sigma^2 \pi^2 T }{a^2}\right) =: \Bf(a),
\end{align*}
where~$\mathrm{E}_1(z):=\int_{z}^{\infty}x^{-1}\E^{-x}\D x$ denotes the exponential integral and~$\mathrm{erf}(\cdot)$ the error function.
Since
$\partial_{z}\mathrm{erf}(z) = \frac{2}{\sqrt{\pi}}\E^{-z^2}$
and
$\partial_{z}\mathrm{E}_1(z) = -z^{-1}\E^{-z}$, then
$\partial_{a}\Bf(a)
= \frac{2a - K\pi}{a^2\pi}K\exp\left\{-\frac{2 \pi^2 \sigma^2 T}{a^2}\right\}$,
and thus~$\Bf(\cdot)$ is decreasing on $(0,\frac{K\pi}{2})$ and increasing on $(\frac{K\pi}{2},\infty)$,
with minimum attained at $\frac{K\pi}{2}$,
so the best upper bound is $\Bf = \Bf(\frac{K\pi}{2})$.
Working with quantities normalised by the spot price,  for example, we obtain for $K=1$, $\sigma=0.2$, $T=1$ and $a=K\pi/2$, then $\Bf \approx 0.848$, while numerical quadrature gives $L^1[\widehat f]\approx0.6544$.

\subsubsection{Black-Scholes}\label{sec:BS}
In the Black-Scholes model, the stock price (assuming no interest rate) is the unique strong solution to the stochastic differential equation
$\D S_{t} = \sigma S_{t}\D W_t$ starting from $S_{0} >0$.
The European Put price with maturity $T>0$ and strike $K>0$ then reads
$$
\Put_{\BS}(S_{0}, K, T, r, \sigma) = K\Nn(-d_{-}) - S_{0}\Nn(-d_{+}),
$$
with~$\Nn(\cdot)$ the Gaussian cumulative distribution function and
$d_{\pm} := \frac{\log(S_{0}/K)}{\sigma\sqrt{T}} \pm \frac{1}{2}\sigma\sqrt{T}$.

\subsubsection*{Truncated Put option in Black-Scholes}\label{ex:TruncatedPutBound}
Fix $\underline{K} \in (0,K)$ and consider the Black-Scholes model
$f(x) = \EE[\Phi_{1}({x+L_T})]$, $x \in \RR$
and $\Phi_{1} = \Phi\circ\exp$
with $\Phi(x) = (K-x)_{+}\one_{\{x\geqslant \underline{K}\}}$. 
Now, 
$$
\|\Phi_{1}\|_{L^1(\RR)}
= \int_{\RR} \left(K-\E^{x}\right)_{+}\one_{\{\E^{x}\geqslant \underline{K}\}} \D x
= \int_{\log(\underline{K})}^{\log(K)} \left(K-\E^{x}\right)\D x
 = K\Big(\log(K)-\log(\underline{K})\Big)
 - (K -\underline{K}),
$$
which is finite.
Corollary~\ref{cor:Levy} with $C=\frac{\sigma^2}{2}$ yields Statement~\ref{st:state} with
$\Bf =  [K(\log(K)-\log(\underline{K})) - (K -\underline{K})]   \frac{1}{\sqrt{2\pi T}\sigma}$. 
We can obtain a sharper bound by computing the Fourier transform $\widehat{\Phi}_{1}(\xi) = \int_{\underline{k}}^{k} (K-\E^x)\, \E^{-2 \I\pi x  \xi} \D x$, 
with $\underline{k}=\log(\underline{K})$ and $k = \log(K)$. We obtain
$$
\widehat{\Phi}_{1}(\xi) = 
\frac{K}{-2 \I\pi \xi} 
\left(\E^{-2 \I\pi k  \xi} - \E^{-2 \I\pi \underline{k} \xi} \right) - \frac{1}{-2 \I\pi \xi +1 } \left(\E^{-(2 \I\pi \xi - 1)k} - \E^{-(2\I\pi \xi - 1) \underline{k} } \right),
$$
hence
$|\widehat{\Phi}_{1}(\xi)| \leqslant \frac{K}{\pi |\xi|}  + \frac{K + \underline{K}}{\sqrt{1+4 \pi^2 \xi^2}} \leqslant \frac{3K+ \underline{K}}{2\pi |\xi|}$
and $|\widehat{\Phi}_{1}(\xi)| \leqslant \|\Phi_{1}\|_{L^1(\RR)}$. Then, for any $a>0$,
\begin{align*}
\LOnefhat  & = \int_{\RR}   |\widehat{\Phi}_{1}(\xi)| |\EE[\E^{2 \pi \I L_T \cdot \xi }] | \D\xi\\
& \leqslant \int_{-1/a}^{1/a} [K(k-\underline{k})-(K-\underline{K})] \E^{-\frac{1}{2}\sigma^2 T \xi^2 (2\pi)^2 }  \D \xi + 2 
\int_{1/a}^\infty  \frac{3K+ \underline{K}}{2\pi |\xi| } \E^{-\frac{1}{2}\sigma^2 T \xi^2 (2\pi)^2}  \D\xi\\
 &=  \frac{ [K(\log(K)-\log(\underline{K}))-(K-\underline{K})]}{\sigma\sqrt{2T\pi}}\mathrm{erf}\left(\frac{2\pi\sigma\sqrt{T}}{a\sqrt{2}}\right) 
 +   \frac{3K+ \underline{K}}{2 \pi}  \mathrm{E}_1\left(
\frac{2 \pi^2 \sigma^2 T }{a^2}\right) =: \Bf(a).
\end{align*}
Similarly to the previous example, $\Bf(\cdot)$ attains its minimum at 
$ a^* =   \frac{2 [K(\log(K)-\log(\underline{K}))-(K-\underline{K})] \pi}{(3K+ \underline{K})}$. For $K=1$, $\underline{K}=0.4$, $\sigma=0.2$ and $T=1$ we obtain $a^*\approx0.5845$ and $\Bf(a^*)\approx0.6283$, slightly below the value $0.6309$ given by Corollary~\ref{cor:Levy}; for comparison, numerical quadrature of $\int_{\RR}|\widehat{\Phi}_1(\xi)|\,\E^{-2\pi^2\sigma^2T\xi^2}\D\xi$ gives $L^1[\widehat f]\approx0.4211$.

\subsubsection*{Put option in Black-Scholes}
For a Put option in a Black-Scholes model, we can then apply the bound for the truncated Put above and combine it with a truncation estimate. Let $f_i(x) = \EE[\Phi_{i}({x+L_T})]$, $x \in \RR$
with $\Phi_{1}$ as in the (truncated) previous example and $\Phi_{2} = \Phi\circ\exp$
with $\Phi(x) = (K-x)_{+}$. Then $f_2$ is the true Put price and~$f_1$ can be approximated using a QNN with a bound given in the previous example, and we can estimate the difference
\[\begin{aligned}
|f_1(x)-f_2(x)| =
\Big| \EE\left[\left(K-\E^{x+L_T}\right)_{+} \one_{\{\E^{x+L_T}\leqslant \underline{K}\}}\right] \Big| 
& \leqslant K \P(L_T < \log(\underline{K})-x) 
\\
& = K \P\left( \Nn(0,1)\leqslant  \frac{ \log(\underline{K})-x + \frac{\sigma^2 T}{2}}{\sigma \sqrt{T}} \right).
\end{aligned}
\]
For example, at the money $x = 0$, for $K=1$, $\underline{K} = 0.4$, $\sigma=0.2$, $T=1$,
the error is about $3.7 \cdot 10^{-6}$.

\begin{remark}[Comparison with the Barron condition]\label{rem:Barron}
	Statement~\ref{st:state} was proved in \cite[Theorem~II.4]{gonon2023universal}  under
	the requirement $\LOnefhat<\infty$, defining the classes~\eqref{eq:Fspaces}. Let us contrast this requirement with the classical Barron condition
	$$
	\Cbarron := \int_{\RR^d}\|\bx\|\,\big|\widehat{f}(\bx)\big|\D\bx<\infty,
	$$
	under which a shallow classical sigmoidal neural network with~$N$ neurons achieves an
	$L^2(\mu)$ error of order $\Cbarron N^{-\half}$~\cite{Barron1993}.
	For $f\in L^{1}(\RR^d)$ we have $\|\widehat{f}\|_{\infty}\leqslant\|f\|_{L^{1}}$, whence
	$$
	\LOnefhat
	= \int_{\|\bx\|\leqslant1}\big|\widehat{f}(\bx)\big|\D\bx
	+ \int_{\|\bx\|>1}\big|\widehat{f}(\bx)\big|\D\bx
	\leqslant \big|B_1(0)\big|\,\|f\|_{L^{1}} + \Cbarron,
	$$
	so that $\Cbarron<\infty$ implies $\LOnefhat<\infty$. The hypothesis on quantum neural networks in	Statement~\ref{st:state}  (cf.\ \cite[Theorem~II.4]{gonon2023universal}) is thus weaker than the sufficient condition for classical sigmoidal neural networks, at the same $\Oo(n^{-\half})$ rate.
	Here we provide two examples showing that the hypothesis is \emph{strictly} weaker.
	Indeed, for $d=1$ and $f(x)=\E^{-|x|}$, one has
	$\widehat{f}(\xi) = 2\left(1+4\pi^2\xi^2\right)^{-1}$, so that
	$$
	\LOnefhat = \int_{\RR}\frac{2\D\xi}{1+4\pi^2\xi^2} = 1<\infty,
	\qquad\text{while}\qquad
	\Cbarron = \int_{\RR}\frac{2|\xi|\D\xi}{1+4\pi^2\xi^2} = +\infty.
	$$
	More generally, in~$\RR^d$, consider the Mat\'ern kernels
	characterised by $\widehat{f}(\bx) = (1+\|\bx\|^{2})^{-s}$. These belong to
	$\Cc(\RR^d)\cap L^{1}(\RR^d)$ for $2s>d$.
	Since $\int_{\RR^d}\|\bx\|^{\beta}(1+\|\bx\|^{2})^{-s}\D\bx$ is finite if and only if
	$2s>d+\beta$, it holds that
	$
	\LOnefhat<\infty$ {and} $
	\Cbarron = +\infty$
	for $s \in (\frac{d}{2},\frac{d+1}{2}]$.
	For such $s$, Statement~\ref{st:state} applies with $\Bf=\LOnefhat$ while the classical neural network bound from 
	Barron \cite{Barron1993} cannot be applied. While other results based on Sobolev-regularity could be applied~\cite{yarotsky2017error}, the constants in these results are hard to track and compute explicitly. 
	
Two caveats are in order. First, the comparison is with Barron's condition for \emph{sigmoidal} networks: since $f^R_{n,\ttheta}(\xx)=\frac1n\sum_{i=1}^nR\cos(\gamma_i)\cos(b_i+\ab_i\cdot\xx)$ is itself a shallow network with cosine activation, the same argument yields the rate $\LOnefhat/\sqrt n$ for classical cosine networks under the same hypothesis. Second, these are sufficient conditions, not lower bounds: the failure of $\Cbarron<\infty$ does not preclude a classical network from approximating~$f$ at rate $N^{-\half}$ by other arguments. The quantum feature is therefore not the function class but the encoding of the $n$ terms in $\nf=\lceil\log_2(4n)\rceil$ qubits.
\end{remark}

\begin{remark}[Short maturities in pure-jump models]\label{rem:VG}
The factorisation $\widehat{f}(\xi)=\widehat{\Phi}_{1}(\xi)\,\EE[\E^{2\I\pi\xi L_T}]$
in Proposition~\ref{prop:convolutional} separates the contribution of the
payoff from that of the model. We use this to give another example when Statement~\ref{st:state} applies, but the classical neural network approximation results~\cite{Barron1993} cannot be applied. 
Take $d=1$, the truncated Put $\Phi(y)=(K-y)_{+}\one_{\{y\geqslant\underline{K}\}}$ of	Section~\ref{sec:BS} and a Variance Gamma process with parameters
$(\sigma,\theta,\nu)$~\cite{MadanCarrChang1998}.
First, the function~$\Phi_{1}$ is continuous at $k=\log(K)$ but has a single jump, of size $K-\underline{K}$, at $\underline{k}=\log(\underline{K})$.
	Using $\E^{k}=K$ in the explicit expression of~$\widehat{\Phi}_{1}$ derived in
	Section~\ref{sec:BS}, 
	one obtains
	$$
	\widehat{\Phi}_{1}(\xi)
	= \frac{K-\underline{K}}{2\I\pi\xi}\,\E^{-2\I\pi\underline{k}\xi}
	+ \Oo\left(|\xi|^{-2}\right),
	\qquad\text{so that}\qquad
	\big|\widehat{\Phi}_{1}(\xi)\big|\sim\frac{K-\underline{K}}{2\pi|\xi|},
	\quad |\xi|\to\infty.
	$$
	Second,
	$|\EE[\E^{\I u L_{T}}]|
	= |\left(1-\I\theta\nu u+\half\sigma^{2}\nu u^{2}\right)^{-T/\nu}|$ and thus
	$\big|\EE[\E^{\I u L_{T}}]\big|\asymp|u|^{-2T/\nu}$ as $|u|\to\infty$.	Combining the two, $\big|\widehat{f}(\xi)\big|\asymp|\xi|^{-(1+2T/\nu)}$, and therefore
	$$
	\LOnefhat<\infty \ \text{ for every } T>0,
	\qquad\text{whereas}\qquad
	\Cbarron<\infty \ \text{ if and only if } \  T>\frac{\nu}{2}.
	$$
	For a typical $\nu=\half$, the Barron bound cannot be applied for small maturities. In contrast, Statement~\ref{st:state} holds with $\Bf=\LOnefhat$ throughout.
Note that in this regime Proposition~\ref{prop:convolutional} and Corollary~\ref{cor:Levy} do not apply, but 
	Statement~\ref{st:state} nevertheless holds due to \cite[Theorem~II.4]{gonon2023universal}.
\end{remark}

\section{Noisy Quantum neural networks}
\label{sec:NoiseModel}
We now develop quantitative Universal Approximation Theorems for noisy QNNs. These results apply both to the Fourier-integrable functions in \eqref{eq:Fspaces} and to expectation functions, as in Section~\ref{subsec:expectations},
and in particular, to the Black-Scholes model. 
We start with the setup from~\cite[Section~2]{gonon2023universal},
where a quantum circuit evolves an initial quantum state to an output one.
To account for noise in quantum hardware, this setup must be translated to a density operator formalism.
In the wavefunction formalism of~\cite{gonon2023universal}, the initial $\nf$-qubit state is $\ket{0}^{\otimes \nf}$. 
In the density operator formalism, this corresponds to the pure state density matrix
$\rho_0 = \ketbra{0}{0}^{\otimes \nf}$, which then evolves to
\begin{equation}\label{eq:rho1_rho2}
\rho_1 = \Vg \rho_0 \Vg^{\dagger}
\qquad\text{and}\qquad
    \rho_2 = \Ug(\ttheta, \xx) \rho_1 \Ug^{\dagger}(\ttheta, \xx),
\end{equation}
for some unitary operators~$\Vg$ and~$\Ug$, 
the latter being a function of parameters~$\ttheta$ and inputs~$\xx$.
To add noise in the system, we require different operations, namely quantum channels, which we now recall briefly and refer the reader to~\cite[Section~3.5]{scherer2019mathematics} for full details.

\subsection{Background on Quantum operations}\label{sec:QuantumChannels}
Let~$\Hh$ be a finite-dimensional Hilbert space
and denote
$\Dd(\Hh):=\Big\{
\rho \in \mathcal{B}(\Hh): \rho^\dagger = \rho,\ \rho\geqslant 0,\ \Tr[\rho] = 1
\Big\}$,
the space of density matrices, with $\mathcal{B}(\Hh)$ the linear operators on~$\Hh$.

\begin{definition}\label{def:QChannel}
A quantum
channel is a linear map $\Psi:\mathcal{B}(\Hh)\to\mathcal{B}(\Hh)$ that is completely positive and trace preserving (CPTP); in particular, it maps~$\Dd(\Hh)$ into itself.
\end{definition}
The \emph{completely positive} property means that any extension
$\Psi\otimes \Ig$ is also a positive map, where~$\Ig$ is the identity operator of any (finite) dimension.
The following result, due to Kraus~\cite{kraus1971general} gives a useful characterisation of quantum channels:
\begin{proposition}
Given quantum channel $\Psi:\Dd(\Hh)\to\Dd(\Hh)$, there exists a sequence of linear maps $\{\Kr_{\ell}\}_{\ell}$, called 
the Kraus operators,
 such that the \emph{operator-sum representation} 
$\Psi(\rho) = \sum_{\ell}\Kr_{\ell}\rho \Kr_{\ell}^\dagger$
holds for all $\rho\in\Dd(\Hh)$,
with $\sum_{\ell}\Kr_{\ell}^\dagger\Kr_{\ell} = \Ig$.
The sums over~$\ell$ run from~$1$ to at most~$\dim(\Hh)^2$.
\end{proposition}

In the case of a single Kraus operator~$\Kr_{1}$, the decomposition reads
$\Psi(\rho) = \Kr_{1}\rho\Kr_{1}^{\dagger}$ with $\Kr_{1}^{\dagger}\Kr_{1} = \Ig$; in finite dimension,
$\Kr_1$ is then unitary, and the channel reduces to the standard (noiseless) quantum evolution.

\subsection{Noisy quantum circuit}
We now consider a general quantum channel~$\Psi$ as in Definition~\ref{def:QChannel} with Kraus representation
$\Psi(\rho) = \sum_{k} \Kr_k \rho \Kr_k^{\dagger}$.
The noisy version of~\eqref{eq:rho1_rho2} reads
\begin{equation}\label{eq:rhoNoisy}
\widetilde{\rho}_{2}
= \Psi_{\Ug}\Big( \Ug(\ttheta, \xx) \Psi_{\Vg}\left(\Vg\rho_0\Vg^\dagger\right) \Ug^{\dagger}(\ttheta, \xx) \Big)
= \Psi_{\Ug}\Big( \Ug(\ttheta, \xx) \Psi_{\Vg}(\rho_1) \Ug^{\dagger}(\ttheta, \xx) \Big)
\end{equation}
where $\Psi_{\Vg}$ and $\Psi_{\Ug}$ represent the noise channels respectively associated with~$\Vg$ and~$\Ug$ and~$\rho_1$ is the noiseless quantum state (as a density operator) after the first layer as in~\eqref{eq:rho1_rho2}.
Using $\ket{\psi_{1}} = \frac{1}{\sqrt{n}} \sum_{i=0}^{n-1} \ket{4i}$ from~\cite{gonon2023universal}, where~$n$ is such that
$\nf = \lceil \log_2(4n) \rceil$,
then
$$
\rho_1 = \ketbra{\psi_{1}}{\psi_{1}} = \frac{1}{n} \sum_{i=0}^{n-1} \sum_{j=0}^{n-1} \ketbra{4i}{4j}.
$$

The general noisy state is therefore
\begin{equation}\label{rho_noisy}
\widetilde{\rho}_{2} = \frac{1}{n} 
\sum_{k} \sum_{\ell} \sum_{i=0}^{n-1} \sum_{j=0}^{n-1} \Kr^{\Ug}_{k} \Ug(\ttheta, \xx) \Kr^{\Vg}_{\ell} \ketbra{4i}{4j} \Kr^{\Vg \dagger}_{\ell} \Ug^{\dagger}(\ttheta, \xx) \Kr^{\Ug \dagger}_{k}.
\end{equation}

\begin{remark}
The noisy density state $\widetilde{\rho}_{2}$ can be re-expressed as a wavefunction through a convex combination of pure states. 
For each pair of Kraus indices $(k, \ell)$, 
the quantum state
$\ket{\tilde{\phi}_{k,\ell}} := \Kr^{\Ug}_{k} \Ug(\ttheta, \xx) \Kr^{\Vg}_{\ell} \ket{\psi_{1}}$,
is not normalised, since
$\braket{\tilde{\phi}_{k,\ell}|\tilde{\phi}_{k,\ell}}$ is not necessarily one, as each Kraus operator is not unitary.
We define its normalised version
$\ket{\phi_{k,\ell}} := \frac{1}{\sqrt{p_{k,\ell}}} \ket{\tilde{\phi}_{k,\ell}}$,
with 
$p_{k,\ell} := \braket{\tilde{\phi}_{k,\ell}|\tilde{\phi}_{k,\ell}}$ and therefore
\begin{equation}\label{eq:kraus_ensemble}
\widetilde{\rho}_{2} = \sum_{k,\ell} \ketbra{\tilde{\phi}_{k,\ell}}{\tilde{\phi}_{k,\ell}} = \sum_{k,\ell} p_{k,\ell} \ketbra{\phi_{k,\ell}}{\phi_{k,\ell}}.
\end{equation}
The completeness relations $\sum_k \Kr^{\Ug\dagger}_k \Kr^{\Ug}_k = \sum_\ell \Kr^{\Vg\dagger}_\ell \Kr^{\Vg}_\ell = \Ig$ guarantee that $\sum_{k,\ell} p_{k,\ell} = 1$. 
\end{remark}

\subsection{Noisy probabilities}\label{sec:NoisyProba}

For measurement outcome $m \in \{0, 1, 2, 3\}$, we denote~$\Pi_m$ 
the projector onto the subspace $\{m, 4+m, \ldots, 4(n-1)+m\}$,
that is $\Pi_m= \sum_{k=0}^{n-1}\ketbra{4k+m}{4k+m}$.

\begin{proposition}\label{prop:kraus_trajectory_meas}
Under general Kraus noise, the quantum measurement probability reads
$$
\widetilde{\PP}_m = \Tr{\Bigl[\Pi_m\widetilde{\rho}_{2}\Bigr]}
 = \sum_{k,\ell} p_{k,\ell} \sum_{j=0}^{n-1} |\braket{4j+m | \phi_{k,\ell}}|^2,
\qquad\text{for each }m \in \{0, 1, 2, 3\}.
$$
\end{proposition}

\begin{proof}
The first equality follows from the Born rule applied to~\eqref{rho_noisy}.
Using~\eqref{eq:kraus_ensemble},
\begin{align*}
\Tr\Bigl[\Pi_m \widetilde{\rho}_{2}\Bigr] = \sum_{k,\ell} p_{k,\ell} \Tr\Bigl[\Pi_m \ket{\phi_{k,\ell}}\bra{\phi_{k,\ell}}\Bigr]  & = \sum_{k,\ell} p_{k,\ell} \braket{\phi_{k,\ell}|\Pi_m|\phi_{k,\ell}} \\
 & = \sum_{k,\ell} p_{k,\ell} \sum_{j=0}^{n-1} |\braket{4j+m | \phi_{k,\ell}}|^2.
\end{align*}
\end{proof}

\begin{remark}
If the noise of~$\Vg$ is trivial ($\Kr^{\Vg}_\ell=\Ig$), then $\ket{\phi_{k}}=p_k^{-1/2}\Kr^{\Ug}_k\Ug(\ttheta,\xx)\ket{\psi_1}$. Writing
$$
\Kr_k^{\Ug}\Ug(\ttheta,\xx)\ket{\psi_1}=\sum_{o=0}^{3}\sum_{i=0}^{2^{\nf-2}-1}\widetilde u_{k,o,i}\ket{4i+o}
$$
in the computational basis, Proposition~\ref{prop:kraus_trajectory_meas} gives
$$
\widetilde\PP_m=\sum_k p_k\sum_{j=0}^{n-1}\big|\braket{4j+m|\phi_k}\big|^2=\sum_k\sum_{j=0}^{n-1}|\widetilde u_{k,m,j}|^2 .
$$
For Kraus operators that are linear combinations of Pauli strings on the $\nf$ qubits, $\Kr_k^{\Ug}$ may move amplitude between blocks and into the padding blocks $i\geqslant n$, which is why the expansion runs up to $i=2^{\nf-2}-1$.
\end{remark}

\subsection{Noisy Quantitative Universal Approximation Theorem}

We now connect the noisy QNN output to the quantitative Universal Approximation Theorem from~\cite{gonon2023universal}.
We refer to Section~\ref{sec:theory} and~\eqref{eq:Fspaces} 
for the necessary tools.
Similarly to using $\,\widetilde{}\,$ for the noisy density operator~\eqref{eq:rhoNoisy}, we keep this notation for the noisy version of the QNN output.

\begin{definition}\label{def:fidelity}
Given two density operators~$\rho$ and~$\sigma$, their \emph{fidelity} ~\cite[Section~9.2.2]{nielsen2010quantum} is
$\Ff[\rho,\sigma] \coloneqq \Tr\left[\left(\rho^{\half}\sigma\rho^{\half}\right)^{\half}\right]$.
\end{definition}

\begin{theorem}[Universal approximation under general CPTP noise]\label{thm:general_noise_bound}
For any $R > 0$, $f \in \Ff_R$, $n \in \NN$ and CPTP channels~$\Psi_{\Vg}$ and~$\Psi_{\Ug}$, there exists $\ttheta \in \Theta$ such that
\begin{equation}\label{eq:general_noise_error}
\begin{split}
\left(\int_{\RR^d} \left|f(\xx) - \widetilde{f}^{R}_{n,\theta}(\xx)\right|^2 \mu(\D\xx)\right)^{\half} &\leqslant \frac{L^1[\widehat{f}]}{\sqrt{n}}\\
&\quad + 2R\left(\int_{\RR^d}\Big(1-\Ff\big[\rho_2(\ttheta,\xx),\widetilde\rho_2(\ttheta,\xx)\big]^2\Big)\mu(\D\xx)\right)^{\half},
\end{split}
\end{equation}
with
$\widetilde{f}^R_{n,\theta}(\cdot) = R\left[1 - 2\left(\widetilde{\PP}_{1} + \widetilde{\PP}_{2}\right)\right]
$ -- mimicking~\eqref{eq:f_QNN_NoNoise} --, $\rho_2$ and~$\widetilde\rho_2$ the noiseless and noisy output states~\eqref{eq:rho1_rho2} and~\eqref{eq:rhoNoisy}, and~$\Ff$ the fidelity of Definition~\ref{def:fidelity}. Since $\rho_2=\ketbra{\psi_2}{\psi_2}$ is pure, with $\ket{\psi_{2}(\ttheta,\xx)}\coloneqq\Ug(\ttheta,\xx)\Vg\ket{0}^{\otimes\nf}$,
\begin{equation}\label{eq:ensemble_fidelity}
\Ff\big[\rho_2,\widetilde\rho_2\big]^2=\braket{\psi_2|\widetilde\rho_2|\psi_2}=\sum_{k,\ell}p_{k,\ell}\,\big|\braket{\psi_2|\phi_{k,\ell}}\big|^2,
\end{equation}
the average over the Kraus trajectories~\eqref{eq:kraus_ensemble} of the squared overlaps with the noiseless state.
\end{theorem}
 
\begin{proof}
By the triangle inequality,
\begin{align}
\Bigl\|f - \widetilde{f}^R_{n,\theta}\Bigr\|_{L^2(\mu)} \leqslant \Bigl\|f - f^R_{n,\theta}\Bigr\|_{L^2(\mu)} + \Bigl\|f^R_{n,\theta} - \widetilde{f}^R_{n,\theta}\Bigr\|_{L^2(\mu)}.
\end{align}
For the first term,
the quantitative quantum Universal Approximation Theorem~\cite[Theorem~II.4]{gonon2023universal} yields that there exists $\ttheta \in \Theta$ such that
$\|f - f^R_{n,\theta}\|_{L^2(\mu)} \leqslant \frac{L^1[\widehat{f}]}{\sqrt{n}}$.
For the second term, $\Pi\coloneqq\Pi_1+\Pi_2$ is an orthogonal projection, so that $0\leqslant\Pi\leqslant\Ig$, and Proposition~\ref{prop:kraus_trajectory_meas} gives, for every~$\xx$,
\begin{align}\label{intermediate_diff}
\left|f^R_{n,\theta}(\xx) - \widetilde{f}^R_{n,\theta}(\xx)\right|
= 2R\left|\Tr\big[\Pi(\rho_2-\widetilde\rho_2)\big]\right|
\leqslant 2R\,D(\rho_2,\widetilde\rho_2),
\end{align}
where $D(\rho,\sigma)\coloneqq\frac12\|\rho-\sigma\|_1$ is the trace distance, which bounds $|\Tr[\Pi(\rho-\sigma)]|$ for every $0\leqslant\Pi\leqslant\Ig$~\cite[Section~9.2.1]{nielsen2010quantum}. The Fuchs-van de Graaf inequality $D(\rho,\sigma)\leqslant\sqrt{1-\Ff[\rho,\sigma]^2}$~\cite[Section~9.2.3]{nielsen2010quantum}, squared and integrated against~$\mu$, yields
\begin{align}\label{bound_noise}
\Bigl\|f^R_{n,\theta} - \widetilde{f}^R_{n,\theta}\Bigr\|_{L^2(\mu)}
\leqslant 2R\left(\int_{\RR^d}\Big(1-\Ff\big[\rho_2(\ttheta,\xx),\widetilde\rho_2(\ttheta,\xx)\big]^2\Big)\mu(\D\xx)\right)^{\half}.
\end{align}
Finally, \eqref{eq:ensemble_fidelity} follows from the expression of the fidelity with a pure state in Section~\ref{sec:depolarisation} and from~\eqref{eq:kraus_ensemble}.
\end{proof}
 
\begin{remark}\label{rem:general_noise}
The noise term of~\eqref{eq:general_noise_error} is at most~$2R$, the trivial bound on $|f^R_{n,\theta}-\widetilde f^R_{n,\theta}|$. For global depolarising noise, Proposition~\ref{prop:fidelity} gives $\Ff^2=\alpha+(1-\alpha)2^{-\nf}$, so that it equals $2R\sqrt{(1-\alpha)(1-2^{-\nf})}$. A bound in terms of the worst trajectory, $\inf_{k,\ell}|\braket{\psi_2|\phi_{k,\ell}}|$, would be vacuous for most physical channels: a Pauli-$\Zg$ Kraus operator on a control qubit in~$\Psi_{\Vg}$, for instance, maps~$\ket{\psi_1}$ to a state orthogonal to it, hence produces a trajectory orthogonal to~$\ket{\psi_2}$ for every $(\ttheta,\xx)$, while the average~\eqref{eq:ensemble_fidelity} only loses the weight of that trajectory. The bound nevertheless remains loose for errors that leave the measured parity unchanged, such as this one, since the fidelity does not see the measurement. For depolarising noise, Theorem~\ref{thm:depolar_approximation} below exploits the structure of the output and is much sharper: its noise contribution is linear in $1-\alpha$, not of order $\sqrt{1-\alpha}$.
\end{remark}
 
\begin{remark}
Theorem~\ref{thm:general_noise_bound} is stated for the parameters~$\ttheta$ provided by the noiseless approximation. Optimising~$\ttheta$ directly on the noisy probabilities~$\widetilde{\PP}_m$ can only lower the empirical risk; for depolarising noise, Corollary~\ref{cor:noise_aware} quantifies the gain. In the hardware experiment of Section~\ref{sec:experiments}, $\ttheta$ is trained on the noiseless model and the resulting circuit is then executed under noise; Section~\ref{subsec:noise_expt} also reports parameters trained against the noisy output, in simulation.
\end{remark}

\subsection{Depolarising noise channel}\label{sec:depolarisation}

Depolarising channels are common CPTP channels to model quantum noise~\cite[Section~8.3.4]{nielsen2010quantum}:

\begin{definition}
The $\df=2^{\nf}$-dimensional depolarising channel is the CPTP map 
$$
\Delta_{\lambda}(\rho) = (1-\lambda)\rho + \frac{\lambda}{\df}\Ig,
$$
with $\nf$ the number of qubits,
and $\lambda \in [0 , 1 + \frac{1}{\df^{2} - 1}]$ ensures complete positivity.
For an orthogonal basis of unitary operators $\{\Ug_i\}_{i=0}^{\df^2-1}$ with $\Ug_0 = \Ig$ and $\Tr[\Ug_i^\dagger \Ug_j] = \df\delta_{ij}$, the general Kraus decomposition reads
$K_0 =  \sqrt{1 - \lambda\frac{\df^2-1}{\df^2}} \Ig$
and 
$K_i = \sqrt{\frac{\lambda}{\df^2}} \Ug_i$ for $i \in \{ 1, \dots, \df^2-1\}$.
\end{definition}

In particular, in the $\nf=1$-qubit case, taking the Pauli  $\{\Ig, \Xg, \Yg, \Zg\}$, 
$$
\Kr_{0} = \sqrt{1-\frac{3\lambda}{4}}\Ig,\qquad
\Kr_{1} = \sqrt{\frac{\lambda}{4}}\Xg,\qquad
\Kr_{2} = \sqrt{\frac{\lambda}{4}}\Yg,\qquad
\Kr_{3} = \sqrt{\frac{\lambda}{4}}\Zg.
$$

\begin{proposition}\label{prop:rho_depol}
Under depolarising noise~$\Delta_{\lambda}$,
the noisy state~\eqref{eq:rhoNoisy} reads
    \begin{equation}
    \label{eq:Depolar_rho}
    \widetilde{\rho}_{2}^{\Delta}
        = (1-\lambda_{\Vg})(1-\lambda_{\Ug})\rho_2
        + \left[ (1-\lambda_{\Ug})\lambda_{\Vg} + \lambda_{\Ug} \right] \frac{\Ig}{2^{\nf}}.
    \end{equation}
\end{proposition}

\begin{proof}
Clearly, from~\cite{gonon2023universal}, both $\Vg$ and $\Ug$ are unitary, so that  $\Ug\Ug^{\dagger} = \Vg\Vg^{\dagger} = \Ig$.
Applying depolarising noise after each gate,  we obtain
\begin{align*}
\widetilde{\rho}_{2}^{\Delta}
= \Delta_{\lambda_{\Ug}}\left( \Ug(\ttheta, \xx) \Delta_{\lambda_{\Vg}}(\rho_1) \Ug^{\dagger}(\ttheta, \xx) \right) 
& = \Delta_{\lambda_{\Ug}}\left( (1-\lambda_{\Vg})\rho_2 + \frac{\lambda_{\Vg}}{2^{\nf}}\Ig \right) \\
& = (1-\lambda_{\Vg})(1-\lambda_{\Ug})\rho_2
+ \left[ (1-\lambda_{\Ug})\lambda_{\Vg} + \lambda_{\Ug} \right] \frac{\Ig}{2^{\nf}}.
    \end{align*}
\end{proof}

The proposition provides us with an explanation of how noise affects the original state.
This then allows us to quantify the error induced by the depolarisation commonly evaluated through the fidelity of Definition \ref{def:fidelity}. When applied to a pure quantum state~$\ket{\psi}$ and some density operator~$\sigma$,
it lets us write
$$
\begin{aligned}
\Ff[\ket{\psi}\bra{\psi},\sigma]
 &= \Tr\left[\left(\ket{\psi}\bra{\psi}\,\sigma\,\ket{\psi}\bra{\psi}\right)^{\half}\right]\\
 &= \Tr\left[\left(\braket{\psi|\sigma|\psi}\,\ket{\psi}\bra{\psi}\right)^{\half}\right]
 = \sqrt{\braket{\psi|\sigma|\psi}}.
\end{aligned}
$$
since $\sqrt{\ket{\psi}\bra{\psi}}=\ket{\psi}\bra{\psi}$. Consider~\eqref{eq:Depolar_rho} and the noiseless $\rho_2 = \Ug(\ttheta, \xx) \Vg\rho_{0}\Vg^{\dagger}\Ug^{\dagger}(\ttheta, \xx)$ in~\eqref{eq:rho1_rho2},
also illustrated in Figure~\ref{fig:Fidelity_Qubits}:

\begin{proposition}\label{prop:fidelity}
Let $\rho_{0} = \ket{0}\bra{0}$ (so~$\rho_2$ is the density operator corresponding to the pure state $\ket{\psi_{2}} = \Ug\Vg\ket{0}$) 
(we drop the $(\ttheta,\xx)$ notation for simplicity).
Then, the fidelity for~$\widetilde{\rho}_{2}^{\Delta}$ is 
\begin{equation}
\Ff\left[\rho_{2},\widetilde{\rho}_{2}^{\Delta}\right] = \left((1-\lambda_{\Vg})(1-\lambda_{\Ug})
+ 
\frac{\left[ (1-\lambda_{\Ug})\lambda_{\Vg} + \lambda_{\Ug} \right]}{2^{\nf}}\right)^{\half}.
    \end{equation}
\end{proposition}

\begin{proof}
    \begin{align*}
\Ff\left[\rho_{2},\widetilde{\rho}_{2}^{\Delta}\right]^2
  = \braket{\psi_{2}|\widetilde{\rho}_{2}^{\Delta}|\psi_{2}}
  & = 
\bra{0}\Vg^{\dagger}\Ug^{\dagger}\widetilde{\rho}_{2}^{\Delta}\Ug \Vg\ket{0}\\
  & = 
\bra{0}\Vg^{\dagger}\Ug^{\dagger}\left((1-\lambda_{\Vg})(1-\lambda_{\Ug})\rho_2
    + \left[ (1-\lambda_{\Ug})\lambda_{\Vg} + \lambda_{\Ug} \right] \frac{\Ig}{2^{\nf}}\right)\Ug \Vg\ket{0}\\
  & = 
(1-\lambda_{\Vg})(1-\lambda_{\Ug})\bra{0}\Vg^{\dagger}\Ug^{\dagger}
\rho_2
\Ug \Vg\ket{0}
 + 
\frac{\left[ (1-\lambda_{\Ug})\lambda_{\Vg} + \lambda_{\Ug} \right]}{2^{\nf}
}
\bra{0}\Vg^{\dagger}\Ug^{\dagger}
\Ug \Vg\ket{0}\\
  & = 
(1-\lambda_{\Vg})(1-\lambda_{\Ug})
+ 
\frac{\left[ (1-\lambda_{\Ug})\lambda_{\Vg} + \lambda_{\Ug} \right]}{2^{\nf}}
.
\end{align*}
\end{proof}

\begin{figure}[H]
    \centering
    \begin{subfigure}[b]{0.45\textwidth}
        \centering
        \includegraphics[scale=0.4]{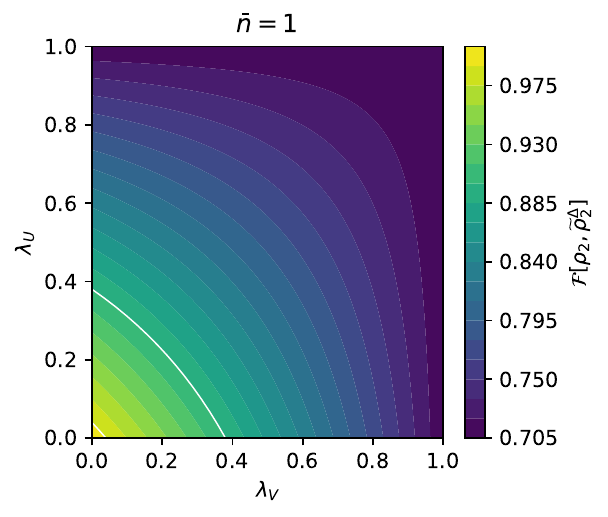}
    \end{subfigure}
    \hfill 
    \begin{subfigure}[b]{0.45\textwidth}
        \centering
        \includegraphics[scale=0.4]{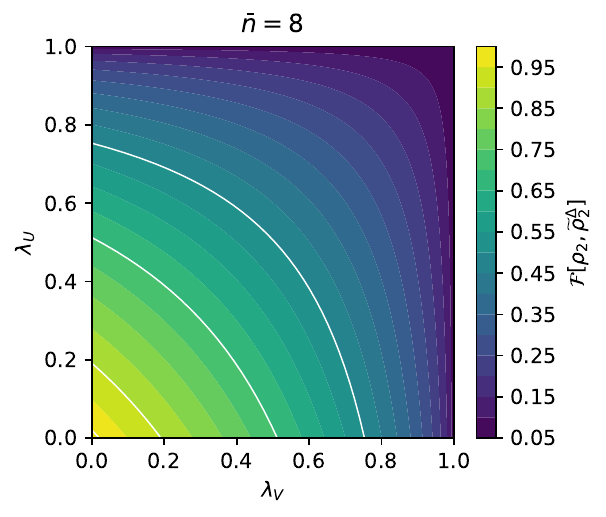}
    \end{subfigure}
\caption{Fidelity computed in Proposition~\ref{prop:fidelity} with $\nf \in \{1, 8\}$ qubits.}    \label{fig:Fidelity_Qubits}
\end{figure}

\begin{proposition}
The probability of outcome $m \in \{0, 1, 2, 3\}$ reads
\begin{equation}\label{eq:noisy_prob}
\widetilde{\PP}_{m} = \alpha \PP_{m} + (1-\alpha) \frac{n}{2^{\nf}},
\qquad\text{where }
\alpha := (1-\lambda_{\Vg})(1-\lambda_{\Ug}),
\end{equation}
where~$\PP_{m}$ is the noiseless probability
and where we recall that $2^{\nf}=4n+n_0$.
\end{proposition}
\begin{proof}
From Proposition~\ref{prop:kraus_trajectory_meas} and Proposition~\ref{prop:rho_depol}, we can write, for each $m\in\{0,1,2,3\}$,
\begin{align*}
\widetilde{\PP}_m = \Tr\Bigl[\Pi_m\widetilde{\rho}_{2}\Bigr]
 & = 
\Tr\Big[\Pi_m
\Big((1-\lambda_{\Vg})(1-\lambda_{\Ug})\rho_2
+ \left[ (1-\lambda_{\Ug})\lambda_{\Vg} + \lambda_{\Ug} \right] \frac{\Ig}{2^{\nf}}
\Big)\Big]\\
 & = 
(1-\lambda_{\Vg})(1-\lambda_{\Ug})\Tr[\Pi_m\rho_2]
+ \left[ (1-\lambda_{\Ug})\lambda_{\Vg} + \lambda_{\Ug} \right] \Tr\left[\Pi_m\frac{\Ig}{2^{\nf}}
\right]\\
 & = 
(1-\lambda_{\Vg})(1-\lambda_{\Ug})\PP_m
+ \frac{1}{2^{\nf}}
\left[ (1-\lambda_{\Ug})\lambda_{\Vg} + \lambda_{\Ug} \right] \Tr[\Pi_m].
\end{align*}
The proposition follows 
if $\Tr[\Pi_m]=n$.
Let~$\mathcal{K}_m:= \{m, 4+m, \dots, 4(n-1)+m\}$
be the target set of measurement outcomes. 
The projection operator~$\Pi_m$ onto the subset of orthonormal computational basis states defined by $\mathcal{K}_m$ can be written explicitly as
\begin{equation}
\Pi_m = \sum_{k \in \mathcal{K}_m} \ket{k}\bra{k} = \sum_{i=0}^{n-1} \ket{4i+m}\bra{4i+m}.
\end{equation}
The trace of an operator is independent of the basis chosen. In the computational basis of the Hilbert space, denoted
$\{\ket{j}\}_{j=0}^{N-1}$ with $N\coloneqq2^{\nf}$, it reads
\begin{align*}
\Tr[\Pi_m]
 & = \sum_{j=0}^{N-1} \bra{j}\Pi_m\ket{j}
 = \sum_{j=0}^{N-1} \bra{j} \left( \sum_{i=0}^{n-1} \ket{4i+m}\bra{4i+m} \right) \ket{j}\\
&  = \sum_{i=0}^{n-1} \sum_{j=0}^{N-1} \braket{j | 4i+m}\braket{4i+m | j}
 = \sum_{i=0}^{n-1} \sum_{j=0}^{N-1} |\braket{j | 4i+m}|^2
  = n.
\end{align*}
\end{proof}

In this depolarising setting, 
the output of the QNN, introduced in Theorem~\ref{thm:general_noise_bound}, is explicit:
\begin{corollary}\label{cor:depolar_qnn_output}
With $f^{R}_{n,\theta}(\cdot)$ the ideal QNN output from~\eqref{eq:f_QNN_NoNoise},
we obtain, under depolarising noise,
\begin{equation}\label{eq:depolar_qnn_explicit}
\widetilde{f}^{R}_{n,\theta}(\xx)
= \alpha f^{R}_{n,\theta}(\xx)
+ R(1 - \alpha)\left(1 - \frac{4n}{2^{\nf}}
\right).
\end{equation}
\end{corollary}

\begin{proof}
Using the definition of $f^{R}_{n,\theta}(\cdot)$ in~\eqref{eq:f_QNN_NoNoise} and~\eqref{eq:noisy_prob}, we have
\begin{align*}
\widetilde{f}^{R}_{n,\theta}(\xx)
 = R - 2R\left(\widetilde{\PP}_{1} + \widetilde{\PP}_{2}\right)
& = R - 2R\left[\alpha(\PP_1 + \PP_2) + 2(1 - \alpha)\frac{n}{2^{\nf}}
\right]\\
& = R - 2R\alpha(\PP_1 + \PP_2) - 4R(1-\alpha)\frac{n}{2^{\nf}}
\\
&  = R - 2R\alpha \frac{R - f^{R}_{n,\theta}(\xx)}{2R} - 4R(1 - \alpha)\frac{n}{2^{\nf}}\\
&  = \alpha f^{R}_{n,\theta}(\xx) + R(1-\alpha)\left(1-\frac{4n}{2^{\nf}}\right).
\end{align*}
\end{proof}

\begin{remark}\label{rem:special}
A useful simplification occurs when $n_{0}=0$, namely when choosing~$n$ and~$\nf$ as
$2^{\nf} = 4n$.
In this case, Corollary~\ref{cor:depolar_qnn_output} reduces to
$\widetilde{f}^{R}_{n,\theta}(\xx)
= \alpha f^{R}_{n,\theta}(\xx)$,
with~$\alpha$ in~\eqref{eq:noisy_prob}.
\end{remark}

The next theorem -- arguably the most immediate result here for practitioners -- shows that a post-processing layer with only two trainable parameters cancels the depolarising bias, so that the noiseless QNN output can be recovered, up to shot noise, without any hardware-level mitigation.

\begin{theorem}[Exact depolarising-noise cancellation by affine post-processing]\label{cor:noise_correction} Consider the QNN $\widetilde{f}^{R}_{n,\bar{\theta}}(\cdot)$ built from the noisy QNN $\widetilde{f}^{R}_{n,\theta}(\cdot)$ by adding  trainable parameters $\beta_1,\beta_2 \in \RR$,
\begin{equation}
    \widetilde{f}^{R}_{n,\bar{\theta}}(\xx)  = \beta_1\, \widetilde{f}^{R}_{n,\theta}(\xx) + \beta_2, \qquad \bar{\theta} = (\theta,\beta_1,\beta_2).
\end{equation}
If $\alpha>0$, then the choice $\beta_1 = \frac{1}{\alpha}$ and  $\beta_2 = -\beta_1 R(1-\alpha)\bigl(1 - \tfrac{4n}{2^{\nf}}\bigr)$ yields $\widetilde{f}^{R}_{n,\bar{\theta}}(\xx) = f^{R}_{n,\theta}(\xx)$ pointwise. Consequently, the class of augmented noisy QNNs contains every noiseless QNN output, so that, with infinitely many shots, the noiseless bound $L^1[\widehat f]/\sqrt n$ of Statement~\ref{st:state} holds under depolarising noise. In the case of Remark~\ref{rem:special} ($2^{\nf}=4n$), $\beta_2 = 0$ and only $\beta_1$ needs to be trained.
\end{theorem}

\begin{proof}
Substituting the specified $(\beta_1,\beta_2)$ into the affine form and
applying Corollary~\ref{cor:depolar_qnn_output}:
\begin{equation}
    \widetilde{f}^{R}_{n,\bar{\theta}}(\xx)  = \beta_1 \widetilde{f}^{R}_{n,\theta}(\xx) + \beta_2 = \beta_1 \alpha f^{R}_{n,\theta}(\xx) + \beta_1 R(1-\alpha)\left(1-\frac{4n}{2^{\nf}}\right) + \beta_2 = f^{R}_{n,\theta}(\xx).
\end{equation}
\end{proof}

\begin{remark}[Practical implication]\label{rem:practical} Theorem~\ref{cor:noise_correction} provides a concrete prescription for noise-aware training that requires no modification to the underlying quantum circuit and adds at most two classical parameters. Rescaling measured expectation values by an estimate of the depolarising fidelity is the basis of the global-depolarising error-mitigation schemes of~\cite{vovrosh2021simple,urbanek2021mitigating}, in which the fidelity is estimated from additional circuits; here, the rescaling is absorbed into two output parameters of the QNN. These parameters can be set \emph{a priori} from the fidelity factor~$\alpha$ computed from calibration data via the formulae of Section~\ref{subsec:hardware_params} (with readout error, $\alpha$ is replaced by $(1-2\pf)^2\alpha$, Corollary~\ref{cor:readout_exact}), or fitted to hardware outputs at points with known targets, such as training points, at the price of additional circuit executions and of the shot cost of Remark~\ref{rem:shot_amplification}. Section~\ref{subsec:hardware_expt} shows that the second option is preferable on current hardware: on \texttt{ibm\_fez}, the calibrated~$\alpha$ is optimistic, and the a priori correction reduces the mean absolute error only from $0.039$ to $0.033$, whereas parameters fitted with~$\theta$ fixed, by leave-one-out cross-validation on the test points, reduce it to $0.018$.
\end{remark}

\begin{remark}[Shot noise]\label{rem:shot_amplification}
The correction multiplies the measured output, and hence its standard error, by~$\beta_1$: by~\eqref{eq:stat_error}, the standard error of the corrected output is $2R\beta_1\sqrt{\widetilde p(1-\widetilde p)/N_{\mathrm{shots}}}$ with $\widetilde p:=\widetilde\PP_1+\widetilde\PP_2$. For $n_0=0$, $\widetilde p-\frac12=\alpha(p-\frac12)$ with $p:=\PP_1+\PP_2$, so that $\widetilde p(1-\widetilde p)\geqslant p(1-p)$, and the number of shots must be increased by a factor of at least $\beta_1^2$ ($\beta_1=1/\alpha$, or $1/\bar\alpha$ with readout error, Corollary~\ref{cor:readout_exact}) to preserve the statistical accuracy of the noiseless circuit.
\end{remark}

We can now provide an accurate bound for the QNN in the depolarising case:
\begin{theorem}[Universal approximation with depolarising noise]\label{thm:depolar_approximation}
For any $R > 0$, $f \in \Ff_R$, $n \in \NN$, and depolarising noise parameters $(\lambda_{\Vg}, \lambda_{\Ug})\in[0,1]^2$, there exists $\ttheta \in \Theta$ such that
\begin{equation}\label{eq:depolar_error_bound}
\left|\int_{\RR^d} \left|f(\xx) - \widetilde{f}^{R}_{n,\theta}(\xx)\right|^2 \mu(\D\xx)\right|^{\half} \leqslant \frac{\alpha L^1[\widehat{f}]}{\sqrt{n}} + (1-\alpha)\|f\|_{L^2(\mu)}
+ R(1-\alpha)\left(1 - \frac{4n}{2^{\nf}}\right).
\end{equation}
\end{theorem}
Note that the last term disappears in the particular case of Remark~\ref{rem:special}.
\begin{proof}
From Corollary~\ref{cor:depolar_qnn_output} and the triangle inequality,
\begin{align*}
\Bigl\|f - \widetilde{f}^{R}_{n,\theta}\|_{L^2(\mu)}
 & = \left\|f - \alpha f^{R}_{n,\theta} - R(1-\alpha) + 4R(1 - \alpha)\frac{n}{2^{\nf}}
\right\|_{L^2(\mu)} \nonumber \\ 
  & \leqslant \Bigl\|f - \alpha f^{R}_{n,\theta}\Bigr\|_{L^2(\mu)} + \left|R(1-\alpha) - 4R(1-\alpha)\frac{n}{2^{\nf}}
\right|\|\mathbf{1}\|_{L^2(\mu)}  \nonumber\\
   & = \|f - \alpha f^{R}_{n,\theta}\|_{L^2(\mu)} + R(1-\alpha)\left(1 - \frac{4n}{2^{\nf}}\right),
\end{align*}
where $\|\mathbf{1}\|_{L^2(\mu)} = 1$ since $\mu$ is a probability measure. The bound on the deviation from the scaled ideal output follows by adding and subtracting $\alpha f$ and using the triangle inequality:
\begin{align}
\Bigl\|f - \alpha f^{R}_{n,\theta}\Bigr\|_{L^2(\mu)} & \leqslant \|f - \alpha f\|_{L^2(\mu)} + \Bigl\|\alpha f - \alpha f^{R}_{n,\theta}\Bigr\|_{L^2(\mu)}\nonumber\\
& = \|f - \alpha f\|_{L^2(\mu)} + \alpha\Bigl\|f - f^{R}_{n,\theta}\Bigr\|_{L^2(\mu)}\nonumber\\
& = (1-\alpha)\|f\|_{L^2(\mu)} + \alpha\Bigl\|f - f^{R}_{n,\theta}\Bigr\|_{L^2(\mu)}.
\end{align}
From the original noiseless UAT, the last term can be bounded as
$\|f - f^{R}_{n,\theta}\|_{L^2(\mu)} \leqslant \frac{L^1[\widehat{f}]}{\sqrt{n}}$,
therefore~\eqref{eq:depolar_error_bound} follows from
$\Bigl\|f - \alpha f^{R}_{n,\theta}\Bigr\|_{L^2(\mu)} \leqslant (1-\alpha)\|f\|_{L^2(\mu)} + \frac{\alpha L^1[\widehat{f}]}{\sqrt{n}}$.
\end{proof}

\begin{remark}\label{rem:three_terms}
The total approximation error has three contributions:
\begin{enumerate}
\item \emph{Approximation}: $\frac{\alpha L^1[\widehat{f}]}{\sqrt{n}}$ is the bound of~\cite{gonon2023universal} scaled by the hardware fidelity factor~$\alpha$; it decreases as the noise increases, because the signal itself is damped.
\item \emph{Systematic}: $(1-\alpha)\|f\|_{L^2(\mu)}$ is the price of this damping; it vanishes as $\alpha\to1$.
\item \emph{Offset}: $R(1-\alpha)\left(1-\frac{4n}{2^{\nf}}\right)=R(1-\alpha)\frac{n_0}{4n+n_0}$ is the constant shift introduced by the maximally mixed component of the noise; it vanishes when $n_0=0$.
\end{enumerate}
The first two terms form the convex combination $\alpha\frac{L^1[\widehat f]}{\sqrt n}+(1-\alpha)\|f\|_{L^2(\mu)}$ of the noiseless bound and of the error $\|f\|_{L^2(\mu)}$ of the trivial predictor~$0$: noise interpolates between the two. In particular, when $n_0=0$, the bound is smaller than $\|f\|_{L^2(\mu)}$, if and only if $n>\big(L^1[\widehat f]/\|f\|_{L^2(\mu)}\big)^2$, independently of $\alpha\in(0,1]$. Below this threshold, the bound increases with~$\alpha$, so that an overestimated~$\alpha$ yields a larger bound, which remains valid. Above it, the bound decreases with~$\alpha$, and an overestimated~$\alpha$ yields a bound that may fail. The same holds for the bound~\eqref{eq:readout_sharp} with readout error, with~$\bar\alpha$ in place of~$\alpha$. The bound of Proposition~\ref{prop:full_noise_approximation} has the same monotonicity in~$\alpha$, but, because of its readout term, it is informative (smaller than $\|f\|_{L^2(\mu)}$) only if $\alpha\big(\|f\|_{L^2(\mu)}-L^1[\widehat f]/\sqrt n\big)>4R\pf$. Finally, this interpolation concerns the parameters of the noiseless approximation executed under noise: by Corollary~\ref{cor:noise_aware}, parameters trained against the noisy output attain the noiseless bound when $n_0=0$ and $\alpha R\geqslant L^1[\widehat f]$.
\end{remark}

\subsection{Hardware-parameterised model for depolarising noise}\label{subsec:hardware_params}

The parameters~$\lambda_{\Vg}$ and~$\lambda_{\Ug}$ are \emph{effective} depolarising parameters determined by the physical circuit implementation; we compute them from gate counts and hardware error rates as follows. \\

\noindent\textbf{Noisy $\Vg$ gate.}
When $n_0=0$ (in particular for $n=8$, used on hardware), $\Vg$ consists of $N_{1\Qg}\coloneqq\nf-2$ Hadamard gates on the control qubits~\cite{gonon2023universal}; for $n_0>0$, it prepares the state $\frac{1}{\sqrt n}\sum_{i=0}^{n-1}\ket{4i}$, and $N_{1\Qg}$ is the number of one-qubit gates of its transpiled form. For independent single-qubit depolarising noise with probability~$\eps_{1\Qg}$ per gate, the effective parameter is $\lambda_{\Vg} = 1 - (1 - \eps_{1\Qg})^{N_{1\Qg}}$.\\

\noindent\textbf{Noisy $\Ug$ gate.}
Let $N_{2\Qg}$ denote the number of two-qubit gates of the transpiled circuit. With the uniformly controlled rotations of Section~\ref{sec:experiments}, $\Ug$ uses $2^{\nf}-4$ CNOTs before routing ($28$ for $n=8$); after transpilation to \texttt{ibm\_fez}, we obtain $N_{2\Qg}=39$, against $718$ for the multi-controlled construction. For a two-qubit error rate~$\eps_{2\Qg}$, the gate contribution is $\lambda_{\Ug}^{\mathrm{gate}} = 1 - (1 - \eps_{2\Qg})^{N_{2\Qg}}$.\\

\noindent\textbf{Decoherence contributions.}
With~$t_{1\Qg}$ and~$t_{2\Qg}$ the durations of one- and two-qubit gates, the circuit execution time $t_{\text{circ}} := \max\{N_{2\Qg} t_{2\Qg}, N_{1\Qg} t_{1\Qg}\}$ leads to the error probabilities
$$
p_{T_1} = 1 - \exp\Big\{-\frac{t_{\text{circ}}}{T_1}\Big\}
\qquad\text{and}\qquad
p_{T_2} = 1 - \exp\Big\{-\frac{t_{\text{circ}}}{T_2}\Big\}
$$
for relaxation time~$T_1$ and dephasing time~$T_2$. Combining both mechanisms, we set
$$
\lambda_{\Ug} \coloneqq 1 - (1 - \lambda_{\Ug}^{\mathrm{gate}})\big(1 - \tfrac{1}{2} p_{T_1}\big)\big(1 - \tfrac{1}{2} p_{T_2}\big).
$$

\noindent\textbf{Combined characterisation.}
Table~\ref{tab:noise_params} reports the noise parameters~$\boldsymbol{\eps}$, namely $\eps_{1\Qg}$, $\eps_{2\Qg}$, $\pf$, $T_1$, $T_2$, $t_{2\Qg}$ and~$N_{2\Qg}$, of the IBM backend \texttt{ibm\_fez} (Heron r2 processor), read from the calibration of the physical qubits and couplers used by the transpiled circuit, as reported by the backend when the job was submitted, together with the resulting $\lambda_{\Vg}=6.9\cdot10^{-4}$, $\lambda_{\Ug}=0.093$ and $\alpha=0.907$. It also reports the noise parameters of the trapped-ion Quantinuum H2 ($56$~qubits) and the modular superconducting Rigetti Cepheus-1-36Q ($36$~qubits) systems, allowing the reader to evaluate the analogous bound on alternative architectures.

\begin{remark}
    We take~$\eps_{1\Qg}$ and~$\eps_{2\Qg}$ to be the error per gate~$r$ reported by the provider, which is obtained by randomised benchmarking. For a $\df$-dimensional depolarising channel with parameter~$\lambda$, the average gate infidelity is $\lambda(\df-1)/\df$, so the depolarising parameter corresponding to~$r$ is $\df r/(\df-1)$, larger by a factor~$2$ for one-qubit and~$4/3$ for two-qubit gates; with this conversion, we would obtain $\alpha=0.883$ instead of $0.907$ on \texttt{ibm\_fez}, which does not affect the conclusions of Section~\ref{subsec:hardware_expt}. Additionally, the recipe ignores the one-qubit gates of~$\Ug$, crosstalk and coherent errors, and it accounts for decoherence on a single qubit. Three further simplifications act in different directions. First, for $n_0=0$, depolarising errors on~$\Vg$ leave the output unchanged (Section~\ref{subsec:noise_expt}), so that~$\lambda_{\Vg}$ is conservative; it is negligible here. Second, $p_{T_2}$ partly double-counts relaxation, since $1/T_2=1/(2T_1)+1/T_\phi$, with~$T_\phi$ the pure-dephasing time. Third, $t_{\text{circ}}$ ignores one-qubit layers, idle and measurement time. Dropping~$\lambda_{\Vg}$ and replacing~$p_{T_2}$ by the pure-dephasing probability gives $\alpha=0.911$ on \texttt{ibm\_fez}, slightly larger than $0.907$; a longer circuit time lowers it. The recipe should therefore be read as an estimate of~$\alpha$; Section~\ref{subsec:hardware_expt} compares it with the contraction measured on hardware.
\end{remark}

\begin{remark}[Equivalence with local gate-level depolarising noise]\label{rem:local_to_global}
While Proposition~\ref{prop:rho_depol} and Corollary~\ref{cor:depolar_qnn_output} are formulated for block-level depolarising channels $\Delta_{\lambda_{\Vg}}$ and $\Delta_{\lambda_{\Ug}}$, the same pure contraction holds \emph{exactly} under local gate-level depolarising noise on the unrouted circuit when $n_0=0$. Specifically, after~$\Vg$, every gate of~$\Ug$ is block-diagonal in the computational basis of the control qubits and acts on only one of the two non-interacting target qubits, while the measured observable is the target parity summed over all~$n$ blocks. Consequently, a local depolarising event on any gate of~$\Ug$ leaves the corresponding target qubit maximally mixed across all blocks, removing that branch's contribution to the parity, whereas a depolarising event on a Hadamard gate in~$\Vg$ preserves the uniform distribution over the $n=2^{\nf-2}$ blocks. Thus, if each gate of~$\Ug$ is followed by a local depolarising channel of rate~$\eps$, the noisy QNN output satisfies $\widetilde{f}^{R}_{n,\theta}(\xx)=(1-\eps)^{N_{\Ug}}f^{R}_{n,\theta}(\xx)$ identically, as verified numerically in Section~\ref{subsec:noise_expt}.
\end{remark}

\subsection{Readout error in the depolarising case}

Readout error is a classical effect occurring after wavefunction collapse and must be modelled separately from quantum noise.

\begin{definition}\label{def:readout}
Let~$\pf$ denote the single-qubit readout error probability
(bit-flip probability during measurement).
For measurement outcome $m \in \{0, 1, 2, 3\}$ encoded in the two target qubits, the measured probability reads
$\overline{\PP}_m := \sum_{m'=0}^{3} q_{m,m'} \widetilde{\PP}_{m'}$,
with~$\widetilde{\PP}_{m'}$ the probability of output~$m'$ in the noisy QNN above; the confusion matrix
$Q = (q_{m,m'} = \pf^{H(m,m')} (1-\pf)^{2 - H(m,m')})$ -- with $H(m,m')$ the Hamming distance between the 2-bit binary representations of~$m$ and~$m'$ -- gives the probabilities of measuring outcome~$m$ when the true quantum state would yield~$m'$.
\end{definition}
The definition models independent, symmetric bit flips of the two target qubits. When $n_0=0$, bit flips of the $\nf-2$ control qubits only permute the blocks~$k$ and leave~$\overline{\PP}_m$ unchanged, so that the definition is exact for independent, symmetric bit flips of all qubits; when $n_0>0$, such flips can also move outcomes into the padding blocks $k\geqslant n$, which the definition neglects.

\begin{proposition}\label{prop:full_noise_approximation}
For depolarising noise with readout error~$\pf$, the approximation bound is,
with $\overline{f}^{R}_{n,\theta} := R\{1-2(\overline{\PP}_1
+\overline{\PP}_2)\}$,
\begin{equation}
\left|\int_{\RR^d} \left|f(\xx) - \overline{f}^{R}_{n,\theta}(\xx)\right|^2 \mu(\D\xx)\right|^{\half} \leqslant \frac{\alpha}{\sqrt{n}}L^1[\widehat{f}]
+ (1-\alpha)\|f\|_{L^2(\mu)} + R(1-\alpha)\left(1 - \frac{4n}{2^{\nf}}\right) + 4R\pf.
\end{equation}
\end{proposition}

\begin{proof}
Write $\widetilde{\PP}_{\mathrm{odd}}:=\widetilde{\PP}_1+\widetilde{\PP}_2$ and $\widetilde{\PP}_{\mathrm{even}}:=\widetilde{\PP}_0+\widetilde{\PP}_3$, and similarly for~$\overline{\PP}$: $\widetilde{\PP}_{\mathrm{odd}}$ is the probability that the two target bits have odd parity. Independent bit flips change this parity if and only if exactly one of the two bits flips, which happens with probability $2\pf(1-\pf)$. Hence
\begin{equation}\label{eq:parity_flip}
\overline{\PP}_{\mathrm{odd}}-\widetilde{\PP}_{\mathrm{odd}}
 =2\pf(1-\pf)\big(\widetilde{\PP}_{\mathrm{even}}-\widetilde{\PP}_{\mathrm{odd}}\big),
\end{equation}
and since $\widetilde{\PP}_{\mathrm{even}}$ and $\widetilde{\PP}_{\mathrm{odd}}$ are non-negative with sum at most one,
$$
\left|\overline{f}^{R}_{n,\theta}(\xx)-\widetilde{f}^{R}_{n,\theta}(\xx)\right|
=2R\left|\overline{\PP}_{\mathrm{odd}}-\widetilde{\PP}_{\mathrm{odd}}\right|
\leqslant 4R\pf(1-\pf)\leqslant 4R\pf
$$
for every~$\xx$, hence in $L^2(\mu)$. The triangle inequality with~\eqref{eq:depolar_error_bound} concludes.
\end{proof}

The parity argument in fact gives an identity: readout error acts as a further contraction of the output.

\begin{corollary}\label{cor:readout_exact}
Under depolarising noise and the readout model of Definition~\ref{def:readout},
\begin{equation}\label{eq:readout_exact}
\overline{f}^{R}_{n,\theta}(\xx)=\bar\alpha\, f^{R}_{n,\theta}(\xx)+R(1-\alpha)\left(1-\frac{4n}{2^{\nf}}\right),
\quad\text{for all }\xx\in\RR^d
\quad\text{with } \bar\alpha\coloneqq(1-2\pf)^2\alpha.
\end{equation}
Consequently, if $\bar\alpha>0$, Theorem~\ref{cor:noise_correction} holds for the measured output~$\overline{f}^{R}_{n,\theta}$ with $\beta_1=1/\bar\alpha$ and $\beta_2=-\beta_1R(1-\alpha)\bigl(1-\frac{4n}{2^{\nf}}\bigr)$, and
\begin{equation}\label{eq:readout_sharp}
\left\|f-\overline{f}^{R}_{n,\theta}\right\|_{L^2(\mu)}\leqslant
\frac{\bar\alpha\,L^1[\widehat f]}{\sqrt n}+(1-\bar\alpha)\|f\|_{L^2(\mu)}+R(1-\alpha)\left(1-\frac{4n}{2^{\nf}}\right),
\end{equation}
which is never larger than the bound of Proposition~\ref{prop:full_noise_approximation}.
\end{corollary}

\begin{proof}
Let $S:=\sum_{m=0}^{3}\widetilde{\PP}_m=\widetilde{\PP}_{\mathrm{even}}+\widetilde{\PP}_{\mathrm{odd}}$. Since the noiseless probabilities sum to one, \eqref{eq:noisy_prob} gives $S=\alpha+(1-\alpha)\frac{4n}{2^{\nf}}$, so that $1-S=(1-\alpha)\bigl(1-\frac{4n}{2^{\nf}}\bigr)$. By~\eqref{eq:parity_flip}, $\overline{\PP}_{\mathrm{odd}}=(1-2\pf)^2\widetilde{\PP}_{\mathrm{odd}}+2\pf(1-\pf)S$, hence
$$
\overline{f}^{R}_{n,\theta}=R\left(1-2\overline{\PP}_{\mathrm{odd}}\right)=(1-2\pf)^2\,\widetilde{f}^{R}_{n,\theta}+4R\pf(1-\pf)(1-S).
$$
Inserting~\eqref{eq:depolar_qnn_explicit} and using $(1-2\pf)^2+4\pf(1-\pf)=1$ yields~\eqref{eq:readout_exact}. The statement on Theorem~\ref{cor:noise_correction} follows as in its proof, and~\eqref{eq:readout_sharp} as in the proof of Theorem~\ref{thm:depolar_approximation}, with~$\alpha$ replaced by~$\bar\alpha$ in the first two terms. Finally, the bound~\eqref{eq:readout_sharp} minus that of Proposition~\ref{prop:full_noise_approximation} equals $(\alpha-\bar\alpha)\big(\|f\|_{L^2(\mu)}-\frac{L^1[\widehat f]}{\sqrt n}\big)-4R\pf$, which is non-positive, since $\alpha-\bar\alpha=4\pf(1-\pf)\alpha\leqslant4\pf$ and $\|f\|_{L^2(\mu)}\leqslant\sup_{\xx}|f(\xx)|\leqslant L^1[\widehat f]\leqslant R$.
\end{proof}

If the two target qubits have different flip probabilities~$p_a$ and~$p_b$, the same computation holds with $(1-2\pf)^2$ replaced by $(1-2p_a)(1-2p_b)$.
The bounds of Theorem~\ref{thm:depolar_approximation}, Proposition~\ref{prop:full_noise_approximation} and Corollary~\ref{cor:readout_exact} hold for the parameters of the noiseless approximation, executed under noise. When $n_0=0$, the noisy output is a rescaled noiseless QNN, which removes the systematic term for parameters trained against the noisy output:
 
\begin{corollary}[Noise-aware parameters]\label{cor:noise_aware}
Let $n_0=0$ and consider depolarising noise with readout error, with $\bar\alpha=(1-2\pf)^2\alpha$ as in Corollary~\ref{cor:readout_exact}. If $f\in\Ff_R$ and $\bar\alpha R\geqslant L^1[\widehat f]$, there exists $\ttheta\in\Theta$ such that
\begin{equation}\label{eq:noise_aware}
\left\|f-\overline{f}^{R}_{n,\theta}\right\|_{L^2(\mu)}\leqslant\frac{L^1[\widehat f]}{\sqrt n}.
\end{equation}
The same holds without readout error, with~$\alpha$ in place of~$\bar\alpha$, and, under Statement~\ref{st:state}, with~$\Bf$ in place of~$L^1[\widehat f]$ on the right-hand side.
\end{corollary}
 
\begin{proof}
For $n_0=0$, \eqref{eq:readout_exact} reads $\overline{f}^{R}_{n,\theta}=\bar\alpha f^{R}_{n,\theta}$. By~\eqref{eq:f_QNN_NoNoise}, $f^{R}_{n,\theta}=R\{1-2(\PP_1+\PP_2)\}$, where the probabilities do not depend on~$R$, so that $\bar\alpha f^{R}_{n,\theta}=f^{\bar\alpha R}_{n,\theta}$. Since $L^1[\widehat f]\leqslant\bar\alpha R$, Statement~\ref{st:state}, applied with~$\bar\alpha R$ in place of~$R$, gives~$\ttheta$ such that $\|f-f^{\bar\alpha R}_{n,\theta}\|_{L^2(\mu)}\leqslant L^1[\widehat f]/\sqrt n$.
\end{proof}
 
\begin{remark}\label{rem:noise_aware}
The parameters of Corollary~\ref{cor:noise_aware} depend on~$\bar\alpha$: they are obtained by training against the noisy output, or, when~$\bar\alpha$ is known, by training the noiseless model with~$R$ replaced by~$\bar\alpha R$. Unlike Theorem~\ref{cor:noise_correction}, no rescaling of the measured output is needed, so that the shot noise is not amplified. The condition $\bar\alpha R\geqslant L^1[\widehat f]$ can always be met by increasing~$R$, at the price of the shot-noise term $R^2/N_{\mathrm{shots}}$ of Remark~\ref{rem:shots}, which grows by the factor~$\bar\alpha^{-2}$ when~$R$ is increased from $L^1[\widehat f]$ to $L^1[\widehat f]/\bar\alpha$. On \texttt{ibm\_fez}, the measured contraction $\bar\alpha\approx0.49$ exceeds $L^1[\widehat f_1]=0.42$ at $R=1$ (Section~\ref{subsec:hardware_expt}), and Section~\ref{subsec:noise_expt} illustrates the corollary in simulation.
\end{remark}

\section{Numerical experiments}
\label{sec:experiments}

We describe the implementation of the quantum circuit and test the bounds of Sections~\ref{sec:theory} and~\ref{sec:NoiseModel} in four settings: the Gaussian density (Section~\ref{subsec:gaussian_expt}), the Black-Scholes Put in the noiseless regime (Section~\ref{subsec:bs_expt}), under simulated gate-level depolarising noise, including an informative regime at $n=32$ (Section~\ref{subsec:noise_expt}), and on the IBM Heron r2 processor \texttt{ibm\_fez} (Section~\ref{subsec:hardware_expt}).
In every experiment, the empirical errors lie within the theoretical bounds. On hardware, the noisy output has the affine form predicted by the depolarising model, but the fidelity factor computed from calibration data overestimates the measured one; we discuss the consequences in Section~\ref{subsec:hardware_expt}.
The experiments are run with the accompanying reproduction notebook, and all numbers below are computed by it (Python~3.13, Qiskit~2.5.2, \texttt{qiskit-aer}~0.17.2, \texttt{qiskit-ibm-runtime}~0.50.0, random seed $20240917$); the hardware numbers are recomputed from the raw counts of the job, which are stored with the notebook.

\url{https://colab.research.google.com/drive/1UdvQQc3eL-NdHzYw6ldCSHxa1quZixDa#scrollTo=db047287}


\subsection{Circuit architecture details}

\subsubsection*{Input normalisation}
Each input coordinate is mapped affinely to $[0,1]$,
\begin{equation}\label{eq:normalisation}
    x \mapsto \frac{x - \underline{x}}{\overline{x} - \underline{x}},
\end{equation}
with $\underline{x}$ and $\overline{x}$ the bounds of the training range. For the Black-Scholes experiments of Sections~\ref{subsec:bs_expt}--\ref{subsec:hardware_expt}, the input is the log-spot $x=\log S$ with $S\in[0.8,1.2]$; for the Gaussian density of Section~\ref{subsec:gaussian_expt}, the input $x\in[-4,4]$ is used. Since both $\LOnefhat$ and the class of QNN outputs are invariant under affine changes of the input variable, the bounds of Section~\ref{sec:theory} apply verbatim to the normalised input.
The scaling parameter is set to
$R = \lceil 1.1 \times \max_{i} \mathrm{P}_i \rceil=1$ for the Put,
with~$\mathrm{P}_i$ the target prices, and to the smallest multiple of $0.1$ above $1.5\max f_\sigma$ for the Gaussian density; in all cases, $R\geqslant\LOnefhat$, as required by Statement~\ref{st:state}.

\subsubsection*{Circuit architecture}

For accuracy parameter $n$ (the number of ``accuracy blocks''), the circuit operates on
$\nf = \lceil \log_2(4n) \rceil$ qubits; unless stated otherwise,
$n = 8$ and $\nf = 5$.
The circuit comprises three stages:
\begin{itemize}
\item \textbf{State preparation ($\Vg$)}: Hadamard gates applied to the $\nf - 2$ control qubits when $n_0=0$ (a state preparation of $\frac{1}{\sqrt n}\sum_{i=0}^{n-1}\ket{4i}$ otherwise), leaving the two target qubits in $\ket{00}$.

\item \textbf{Parameterised unitary ($\Ug$)}: Block-diagonal unitary implementing $n$ accuracy blocks, each contributing a term to the Fourier-like approximation.

\item \textbf{Measurement}: Computational basis measurement.
\end{itemize}

\subsubsection*{Uniformly Controlled Rotation Decomposition}
The parameterised unitary~$\Ug(\ttheta)$ introduced in~\cite{gonon2023universal}
has a block-diagonal structure
\begin{equation}
    \Ug(\ttheta,\xx) = \sum_{k=0}^{n-1} \ketbra{k}{k} \otimes \overline{\Ug}^{(k)}(\ttheta^{k},\xx),
\end{equation}
with $\ttheta^k = (\ab^k,b^k,\gamma^k)$
and where each $\overline{\Ug}^{(k)}(\ttheta^{k},\xx) =
\Ug_1^{(k)}(\ab^{k},b^{k},\xx) \otimes \Ug_2^{(k)}(\gamma^{k})$ acts on the two target qubits.
A naive implementation applies one multi-controlled two-qubit gate per block. \emph{Uniformly controlled rotations} (UCR)~\cite{mottonen2004} instead decompose~$\Ug$ into
one UCR$_z$ on qubit~0, conjugated by Hadamard gates,
and one UCR$_y$ on qubit~1, each controlled by the $\nf-2$ control qubits. Our recursive implementation uses $2^{\nf}-4$ CNOTs ($28$ for $n=8$; the Gray-code ordering of~\cite{mottonen2004} would reduce this to $2^{\nf-1}$). After transpilation to \texttt{ibm\_fez} (native CZ gates, heavy-hex connectivity), the circuit contains $39$ two-qubit gates, against $718$ for the multi-controlled construction. This reduction is critical for hardware execution, where two-qubit gate errors dominate (Section~\ref{subsec:hardware_params}).

\subsubsection*{Parameter optimisation}
Methods~A and~B are trained on the closed-form output of the QNN~\cite[Proposition~VI.1]{gonon2023universal},
\begin{equation}\label{eq:closed_form}
  f^{R}_{n,\ttheta}(\xx)
  = \frac{1}{n}\sum_{i=1}^{n} R\cos(\gamma_i)\cos(b_i + \ab_i \cdot \xx),
\end{equation}
which the circuit reproduces exactly (see \emph{Circuit validation} below). We focus on $d=1$ in our numerical experiments so that global optimisation over the $n(d+2)$ parameters can be solved to very high accuracy, isolating quantum hardware and shot noise from classical training error. We implement three parameter optimisation strategies:

\begin{center}
\begin{tabular}{clp{6.6cm}}
\toprule
\textbf{Method} & \textbf{Algorithm} & \textbf{Description} \\
\midrule
A & Differential evolution & Global search over all $n(d+2)$ parameters of~\eqref{eq:closed_form}, followed by L-BFGS-B polishing \\
B & Last block & Blocks $1,\dots,n-1$ frozen at a random initialisation; only $(\ab^{n},b^{n},\gamma^{n})$ optimised (differential evolution) \\
C & Random features & Frozen random $(\mathbf A,\mathbf B)$ and the single-target circuit of~\cite[Section~III]{gonon2023universal}: $n=32$ features on $\lceil\log_2(2n)\rceil=6$ qubits, estimated with $8192$ shots; ridge regression (standardised features, penalty~$1$, with intercept) of the targets on the features $2\PP_{2j}-\frac1n$ \\
\bottomrule
\end{tabular}
\end{center}

For Methods~A and~B, the loss function is the mean-squared error over the training set,
\begin{equation}
\mathcal{L}(\theta) := \frac{1}{N_{\mathrm{train}}}\sum_{i=1}^{N_{\mathrm{train}}}
 \Big( f_n^R(x_i; \theta) - \mathrm{P}_i \Big)^2,
\end{equation}
with $N_{\mathrm{train}}=500$ for the Put and $100$ to $200$ for the Gaussian density, and each block's parameters $(\ab^k,b^k,\gamma^k)$ range over $[-2\pi,2\pi]^d\times[-\pi,\pi]\times[0,2\pi]$ for the Put and $[-10\pi,10\pi]\times[-2\pi,2\pi]\times[0,2\pi]$ for the Gaussian density. Method~A is used unless stated otherwise. Method~C uses $\lceil\log_2(2n)\rceil$ rather than $\lceil\log_2(4n)\rceil$ qubits, and its~$32$ features are not directly comparable with the $n=8$ blocks of Methods~A and~B; with standardised features and an intercept, the shift $-\frac1n$ plays no role.

\subsubsection*{Measurement}
Outcomes are grouped to extract the probabilities $\{\PP_m\}_{m=0,\cdots,3}$.
For each outcome $o \in \{0, \ldots, 2^{\nf}-1\}$,
we write
$o = 4k + m$
with $k \in \{0, \ldots, 2^{\nf-2}-1\}$, $m \in \{0, 1, 2, 3\}$.
Only outcomes with $k < n$ contribute (states with $k \geqslant n$ arise from padding to a power-of-two dimension). The probability $\PP_m$ accumulates counts from all valid $o \equiv m \pmod{4}$:
\begin{equation}
    \PP_m = \frac{1}{N_{\text{shots}}} \sum_{\substack{o : o \bmod 4 = m \\ o < 4n}} c_o,
\end{equation}
where $c_{o}$ denotes the count for outcome~$o$.
The output of the QNN is then computed as (see, for example, Corollary~\ref{cor:depolar_qnn_output})
$f_n^R(\xx) = R[1 - 2(\PP_1 + \PP_2)]$.

\subsubsection*{Hardware execution}

The circuits are run on \texttt{ibm\_fez} using the V2 \texttt{Sampler} primitive of Qiskit Runtime, requiring \emph{Instruction Set Architecture} (ISA) circuits. Key steps included:

\begin{itemize}
\item \textbf{Transpilation}: Circuits are transpiled once, with \texttt{optimization\_level=3} and a fixed seed, mapping abstract gates to the backend's native gate set; the same ISA circuits are executed and used to count~$N_{2\Qg}$ and to read the calibration of the physical qubits and couplers they use. Both are saved at submission, so that the analysis of a retrieved job uses the executed circuits and the submission-time calibration.
\item \textbf{Batching}: All test circuits submitted as a single job to minimise queue overhead.
\item \textbf{Shot budget}: $N_{\text{shots}} = 8192$ per circuit.
\item \textbf{No error mitigation}: the sampler is used with its default options, without Pauli twirling, dynamical decoupling or readout-error mitigation, so that the measured output reflects the raw noise of the device.
\end{itemize}

The statistical error on the output price is
\begin{equation}\label{eq:stat_error}
2R \sqrt{\frac{(\PP_1 + \PP_2)(1 - \PP_1 - \PP_2)}{N_{\text{shots}}}}.
\end{equation}

\begin{remark}\label{rem:shots}
The shot estimator of $f^R_{n,\theta}(\xx)$ is unbiased with variance at most $R^2/N_{\mathrm{shots}}$. Hence, under Statement~\ref{st:state}, the root-mean-square $L^2(\mu)$ error of the estimated output is at most $\eps_{n,N}:=\sqrt{\Bf^2/n+R^2/N_{\mathrm{shots}}}$, to which $\delta_{\mathrm{trunc}}$ of~\eqref{eq:bs_bound} is added for the Black-Scholes Put, and the same holds with $\Bf/\sqrt n$ replaced by $\eps_{\mathrm{total}}$ in the noisy case.
\end{remark}

Experiments are implemented in Python with Qiskit and executed either on simulators (\texttt{qiskit-aer}) or on \texttt{ibm\_fez}, whose noise parameters are given in Table~\ref{tab:noise_params}.


\subsubsection*{Circuit validation}
\label{subsec:circuit_validation}
For random parameters $(\ttheta, \xx)$ and $n\in\{2,4,8\}$ ($24$ circuits), the output sampled with $8192$ shots is compared with the closed form~\eqref{eq:closed_form}: normalised by the standard error~\eqref{eq:stat_error}, the residuals satisfy $|z|\leqslant2.09$, with root mean square $0.86$, as expected from shot noise. The probabilities of the simulated state reproduce~\eqref{eq:closed_form} to within $10^{-10}$ for $n\in\{1,2,3,4,6,8,12\}$, including $n_0>0$, and the UCR circuit agrees with the multi-controlled construction: $|\langle\psi_{\mathrm{UCR}}|\psi_{\mathrm{naive}}\rangle|\geqslant1-10^{-13}$ for $n\in\{2,4,6,8\}$, $d\in\{1,5\}$ and parameter scales $\{0.1,1,5\}$ (three random parameter sets each, $72$ tests).


\subsection{Gaussian density approximation}
\label{subsec:gaussian_expt}

We first consider Statement~\ref{st:state} in dimension $d=1$, with $L\equiv0$ and $\Phi_1=f_\sigma$, so that $f=f_\sigma$ as in Section~\ref{sec:GaussianDensity}, and with~$\mu$ the uniform probability measure on $200$ equally spaced points in $[-4,4]$ (Statement~\ref{st:state} requires a probability measure).
The QNN with~$n$ accuracy blocks then approximates~$f_{\sigma}$ with error
$\eps_n := \frac{L^1[\widehat{f}_{\sigma}]}{\sqrt{n}} = \frac{1}{\sigma\sqrt{2\pi n}}$,
hence, for a given approximation accuracy $\eps > 0$, the required number of accuracy blocks is
\begin{equation}\label{eq:n_vs_epsilon}
n \geqslant \frac{1}{2\pi\sigma^2 \eps^2}.
\end{equation}
Considering the case $n_0=0$, i.e. $4n = 2^{\nf}$, this can be recast in the minimal number of qubits required given accuracy~$\eps$, namely (see also Figure~\ref{fig:Qubit_accuracy})
$\nf \geqslant \frac{1}{\log(2)}\log\left(\frac{2}{\pi\sigma^2 \eps^2}\right)$.
\begin{figure}[h]
\centering
\includegraphics[scale=0.5]{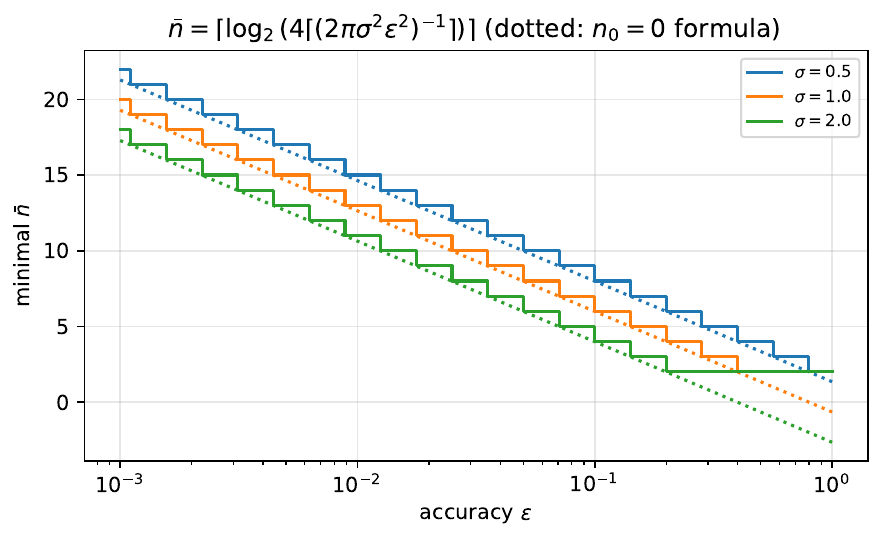}
\caption{Minimal number of qubits $\nf=\lceil\log_2(4\lceil(2\pi\sigma^2\eps^2)^{-1}\rceil)\rceil$ for accuracy~$\eps$ (steps) and the continuous expression $\log_2\big(2/(\pi\sigma^2\eps^2)\big)$ (dotted), for $\sigma\in\{0.5,1,2\}$.}
\label{fig:Qubit_accuracy}
\end{figure}
This highlights the quadratic dependence $n = \Oo(\eps^{-2})$, typical of Monte Carlo-type approximation methods. The dependence on~$\sigma^{-2}$ indicates that narrower Gaussians (smaller~$\sigma$) require more accuracy blocks, consistent with the broader frequency content of their Fourier transforms.

In this section, the QNN output is evaluated with the closed form~\eqref{eq:closed_form}, that is, without shot noise. The parameters are optimised by Method~A on $100$ equally spaced training points, with $R=0.6$, $1.2$ and $0.3$ for $\sigma=1$, $0.5$ and $2$.
Figure~\ref{fig:gaussian_part1} displays the approximation~$f^R_{n,\theta}$ against the true density for $(\sigma,n)= (1,8)$, with the pointwise absolute error on a logarithmic scale: the root-mean-square error (RMSE) is $8.9\cdot10^{-3}$, against the bound $L^1[\widehat{f}_\sigma]/\sqrt{n}=0.141$ of Statement~\ref{st:state}.
Figure~\ref{fig:gaussian_convergence} plots the RMSE and the mean absolute error (MAE) against $n \in \{2, 4, 8, 12, 16\}$, together with the bound (here with $200$ training points and up to three restarts of the optimiser, so that the fit for $n=8$ differs from that of Figure~\ref{fig:gaussian_part1}): the ratio RMSE$/\eps_n$ stays below $0.16$ and MAE$/\eps_n$ below $0.1$. The empirical error does not decrease monotonically in~$n$, because the training is a non-convex global optimisation (the fit for $n=12$ is the least accurate); Statement~\ref{st:state} asserts the existence of parameters attaining the bound, not that an optimiser finds the best ones.
Figure~\ref{fig:multi_sigma} compares $\sigma \in \{0.5, 1, 2\}$ for $n = 8$ ($150$ training points): the RMSE ($0.031$, $0.0090$ and $0.0056$) grows as~$\sigma$ decreases, as does the bound $L^1[\widehat{f}_\sigma]/\sqrt n$ ($0.282$, $0.141$ and $0.071$), and stays below it.

\subsection{Black-Scholes Put option pricing: noiseless regime}
\label{subsec:bs_expt}
We consider the setting of the truncated-put example of Section~\ref{sec:BS}: $K=1$, $\sigma=0.2$, $T=1$, $r=0$, and spot prices $S\in[0.8,1.2]$ (prices normalised by the strike). The QNN input is the log-spot $x=\log S$, rescaled to $[0,1]$ by~\eqref{eq:normalisation}. With $\underline K=0.4$, the true Put~$f_2$ satisfies
\begin{equation}\label{eq:bs_bound}
\big\|f_2-f^R_{n,\theta}\big\|_{L^2(\mu)}\leqslant\frac{\Bf(a^*)}{\sqrt n}+\delta_{\mathrm{trunc}},
\end{equation}
with the truncation error of Section~\ref{sec:BS}
$$
\delta_{\mathrm{trunc}}:=\sup_{S\in[0.8,1.2]}K\,\Nn\left(\frac{\log\underline K-\log S+\sigma^2T/2}{\sigma\sqrt T}\right),
$$
where $\Bf(a^*)=0.6283$, $\delta_{\mathrm{trunc}}=3.8\cdot10^{-4}$ (attained at $S=0.8$) and $R=1\geqslant\Bf(a^*)\geqslant L^1[\widehat{f}_1]$ as required. For $n=8$, this gives $\eps_n:=\Bf(a^*)/\sqrt n=0.2221$ and, with $N_{\mathrm{shots}}=8192$, $\eps_{n,N}:=\delta_{\mathrm{trunc}}+\sqrt{\eps_n^2+R^2/N_{\mathrm{shots}}}=0.2228$ (Remark~\ref{rem:shots}). We train on $500$ spot prices drawn uniformly in $[0.8,1.2]$ and test on $100$ further ones; $\mu$ is the empirical measure of the test points.

\begin{table}[h]\centering\small
\begin{tabular}{cccc}
\toprule
Method & RMSE ($N_{\mathrm{shots}}=8192$) & MAE ($N_{\mathrm{shots}}=8192$) & RMSE ($N_{\mathrm{shots}}=\infty$) \\
\midrule
A & $0.0102$ & $0.0085$ & $2.7\cdot10^{-4}$\\
B & $0.160$ & $0.142$ & $0.159$\\
C & $0.0033$ & $0.0025$ & --\\
\bottomrule
\end{tabular}
\caption{Noiseless Black-Scholes Put ($n=8$, $100$ test points): errors with respect to the Black-Scholes price, with $8192$ shots and with the closed-form output~\eqref{eq:closed_form} ($N_{\mathrm{shots}}=\infty$). Bounds: $\eps_n+\delta_{\mathrm{trunc}}=0.2225$ for $N_{\mathrm{shots}}=\infty$ and $\eps_{n,N}=0.2228$ for $N_{\mathrm{shots}}=8192$.}
\label{tab:bs_methods}
\end{table}

Table~\ref{tab:bs_methods} summarises the three methods. With $R=1$, the shot noise~\eqref{eq:stat_error} has a standard deviation of about $0.011$ and dominates the error of Method~A, whose closed-form output is accurate to $2.7\cdot10^{-4}$. Method~B is worse than the trivial predictor~$0$, whose RMSE is $\|f_2\|_{L^2(\mu)}=0.103$: optimising the last block alone cannot compensate for the $n-1$ random ones. Method~C attains $0.0033$. Figure~\ref{fig:bs_curve} shows the price curve with the envelope $\pm\Bf(a^*)/\sqrt n$ and the pointwise error, and Figure~\ref{fig:method_compare_part1} the per-point errors of Methods~A and~B.
The errors of Methods~A and~C are far below the bound, but the bound itself exceeds $\|f_2\|_{L^2(\mu)}=0.103$: at $n=8$, it is weaker than the trivial predictor~$0$, and it becomes informative only for $n>(\Bf(a^*)/\|f_2\|_{L^2(\mu)})^2$ (Remark~\ref{rem:three_terms}), that is, $n\geqslant38$ with $\Bf(a^*)$ and $n\geqslant17$ with the exact value $L^1[\widehat{f}_1]=0.4211$, for which the bound at $n=8$ is $0.149$. Section~\ref{subsec:noise_expt} considers $n=32$.

Figure~\ref{fig:bs_scaling} repeats the experiment for $n\in\{2,4,6,8\}$ ($n=6$ has $n_0=8$ and uses the state preparation~$\Vg$), with the best of three restarts and a larger optimisation budget, so that its fit for $n=8$ differs from that of Table~\ref{tab:bs_methods}. With the closed-form output, the RMSE lies between $8.0\cdot10^{-5}$ and $4.0\cdot10^{-4}$, between $600$ and $5500$ times smaller than $\eps_n\in[0.22,0.44]$; with $8192$ shots, it lies between $0.0101$ and $0.0115$, at the shot-noise floor, and RMSE$/\eps_{n,N}\leqslant0.05$. The bound thus holds at every~$n$, but it is very conservative for this smooth target, and the empirical error does not follow the rate $n^{-1/2}$, which is only an upper bound.

\subsection{Noise simulation}
\label{subsec:noise_expt}
We now add gate-level depolarising noise to the circuit, simulated with \texttt{AerSimulator}. The circuit is transpiled to the gates \texttt{h}, \texttt{rz}, \texttt{ry} and \texttt{cx} ($N_{\Vg}=3$ gates in~$\Vg$ and $N_{\Ug}=46$ in~$\Ug$ for $n=8$), and every gate is followed by a depolarising channel with parameter $\eps\in\{0.001,0.005,0.01,0.02\}$. Applying the recipe of Section~\ref{subsec:hardware_params} to all gates gives $\alpha=(1-\eps)^{N_{\Vg}+N_{\Ug}}$, the weight of the branch in which no gate is depolarised. Gate-level noise is not globally depolarising, so Corollary~\ref{cor:depolar_qnn_output} need not hold; Table~\ref{tab:noise_sim} compares it with exact density-matrix simulations on $20$ test points.

\begin{table}[h]\centering\small
\begin{tabular}{cccccc}\toprule
$\eps$ & $(1-\eps)^{N_{\Vg}+N_{\Ug}}$ & $\hat\alpha$ (fit) & $\max|\widetilde f-\alpha f^R|$ & RMSE & bound \\\midrule
$0.001$ & $0.9522$ & $0.9550$ & $6.0\cdot10^{-4}$ & $0.0129$ & $0.2166$ \\
$0.005$ & $0.7822$ & $0.7941$ & $2.5\cdot10^{-3}$ & $0.0241$ & $0.1946$ \\
$0.01$ & $0.6111$ & $0.6298$ & $3.9\cdot10^{-3}$ & $0.0395$ & $0.1724$ \\
$0.02$ & $0.3716$ & $0.3948$ & $4.9\cdot10^{-3}$ & $0.0610$ & $0.1414$ \\
\bottomrule\end{tabular}
\caption{Gate-level depolarising noise ($n=8$, $20$ test points). $\hat\alpha$: least-squares contraction factor of the exact noisy output, which equals $\hat\alpha f^R_{n,\theta}$ to within $6\cdot10^{-15}$; $\max|\widetilde f-\alpha f^R|$: deviation of the exact noisy output from Corollary~\ref{cor:depolar_qnn_output} with $\alpha=(1-\eps)^{N_{\Vg}+N_{\Ug}}$; RMSE: error with respect to the Black-Scholes price with $8192$ shots; bound: Theorem~\ref{thm:depolar_approximation} with this~$\alpha$, plus $\delta_{\mathrm{trunc}}$ and shot noise (Remark~\ref{rem:shots}); with the exact~$\hat\alpha$ instead, the bound is larger by at most $0.003$.}
\label{tab:noise_sim}
\end{table}

The noisy output is a pure contraction $\hat\alpha f^R_{n,\theta}$ of the noiseless one, as Corollary~\ref{cor:depolar_qnn_output} predicts for $n_0=0$, and in fact $\hat\alpha=(1-\eps)^{N_{\Ug}}$ to within $10^{-13}$. This has a simple explanation. After~$\Vg$, every gate of the circuit is block-diagonal in the computational basis of the control qubits, which act only as controls; every gate of~$\Ug$ acts on one of the two target qubits (\texttt{h} and \texttt{rz} on qubit~$0$, \texttt{ry} on qubit~$1$, and \texttt{cx} from a control to a target); the two targets never interact; and the measured quantity is their parity, summed over all blocks. A depolarisation event on a gate of~$\Ug$ therefore leaves one target qubit maximally mixed in every block, which removes the contribution of that branch to the parity, whereas a depolarisation event on a Hadamard gate of~$\Vg$ leaves the control register uniformly distributed over the $n$ blocks and has no effect when $n_0=0$. Counting all gates, the recipe of Section~\ref{subsec:hardware_params} is therefore exact for this noise model, up to the factor $(1-\eps)^{N_{\Vg}}$, which slightly overestimates the effect of the noise. These structural properties need not survive routing on hardware. The error with respect to the Black-Scholes price stays below the bound at every noise level; it grows with~$\eps$ essentially through the systematic term $(1-\hat\alpha)\|f\|_{L^2(\mu)}$, which increases from $0.004$ to $0.056$, on top of the shot noise (Figure~\ref{fig:noise_sim}). Since $n=8$ lies below the threshold of Remark~\ref{rem:three_terms}, the bound decreases with~$\eps$ (from $0.217$ to $0.141$) while the error increases (from $0.013$ to $0.061$); both approach $\|f\|_{L^2(\mu)}=0.092$ as $\alpha\to0$.
 
\subsubsection*{An informative regime}
At $n=8$, none of the bounds is informative. With the exact constant $L^1[\widehat f_1]=0.4211$, the noiseless bound beats the trivial predictor for $n>(0.4211/0.103)^2\approx17$, so we repeat the experiment at $n=32$ ($\nf=7$, $n_0=0$) in simulation, with the closed-form output ($N_{\mathrm{shots}}=\infty$) on the $100$ test points of Section~\ref{subsec:bs_expt}. The noiseless bound is then $L^1[\widehat f_1]/\sqrt{32}+\delta_{\mathrm{trunc}}=0.0748$, below $\|f_2\|_{L^2(\mu)}=0.103$ (with~$\Bf(a^*)$ it would be $0.111$). The exact contraction found above persists: with $\eps=0.002$, the density-matrix output equals $(1-\eps)^{N_{\Ug}}f^R_{n,\theta}$, with $N_{\Ug}=190$, to within $2.0\cdot10^{-14}$ on four test points. We may therefore model the noisy output as $\alpha f^R_{n,\theta}$ and compare, for several~$\alpha$, the parameters of the noiseless approximation executed under noise (fixed~$\ttheta$) with parameters trained against the noisy output (noise-aware~$\ttheta$), see Table~\ref{tab:informative}.
 
\begin{table}[h]\centering\small
\begin{tabular}{ccccc}\toprule
$\alpha$ & RMSE, fixed $\ttheta$ & bound~\eqref{eq:depolar_error_bound} & RMSE, noise-aware $\ttheta$ & bound~\eqref{eq:noise_aware} \\\midrule
$1.0$ & $0.0015$ & $0.0748$ & $0.0015$ & $0.0748$ \\
$0.9$ & $0.0099$ & $0.0776$ & $0.0015$ & $0.0748$ \\
$0.7$ & $0.0305$ & $0.0833$ & $0.0016$ & $0.0748$ \\
$0.5$ & $0.0511$ & $0.0889$ & $0.0032$ & $0.0748$ \\
\bottomrule\end{tabular}
\caption{Informative regime ($n=32$, $100$ test points, $N_{\mathrm{shots}}=\infty$, exact constant $L^1[\widehat f_1]=0.4211$, both bounds including $\delta_{\mathrm{trunc}}$). Fixed $\ttheta$: noiseless optimum, output $\alpha f^R_{n,\theta}$; noise-aware $\ttheta$: trained against $\alpha f^R_{n,\theta}$.}
\label{tab:informative}
\end{table}
 
Above the threshold of Remark~\ref{rem:three_terms}, the bound of Theorem~\ref{thm:depolar_approximation} now decreases with~$\alpha$: evaluated with an overestimated $\alpha=0.9$ when the true value is $0.5$, it gives $0.0776$ instead of $0.0889$, still above the error $0.0511$ in this example, but no longer guaranteed. The error of the fixed parameters is dominated by the systematic term $(1-\alpha)\|f_2\|_{L^2(\mu)}$. Noise-aware parameters remove it: their error stays at the noiseless level for every $\alpha\geqslant L^1[\widehat f_1]/R=0.42$, as Corollary~\ref{cor:noise_aware} predicts.

\subsection{Hardware execution on \texttt{ibm\_fez}}
\label{subsec:hardware_expt}

Circuits trained by Method~A are run on \texttt{ibm\_fez} (job \texttt{db2mbae8v0ts73c33ab0}, 6~October 2026) on $10$ test points of Section~\ref{subsec:bs_expt}, with $N_{\mathrm{shots}}=8192$ and the settings of Section~\ref{sec:experiments}.
The transpiled circuits act on the physical qubits $2$, $3$, $4$, $16$ and $23$ (targets $4$ and $3$) and contain $N_{2\Qg}=39$ CZ gates. The calibration of these qubits and couplers reported when the job was submitted (Table~\ref{tab:noise_params}) gives $\eps_{1\Qg}=2.30\cdot10^{-4}$, $\eps_{2\Qg}=2.01\cdot10^{-3}$, $\pf=0.0134$, $T_1=150.3\,\mu$s and $T_2=131.9\,\mu$s, hence, by Section~\ref{subsec:hardware_params}, $\lambda_{\Vg}=6.9\cdot10^{-4}$, $\lambda_{\Ug}=0.093$ and $\alpha=0.907$.

\begin{figure}[tbp]
\centering
\includegraphics[width=\linewidth]{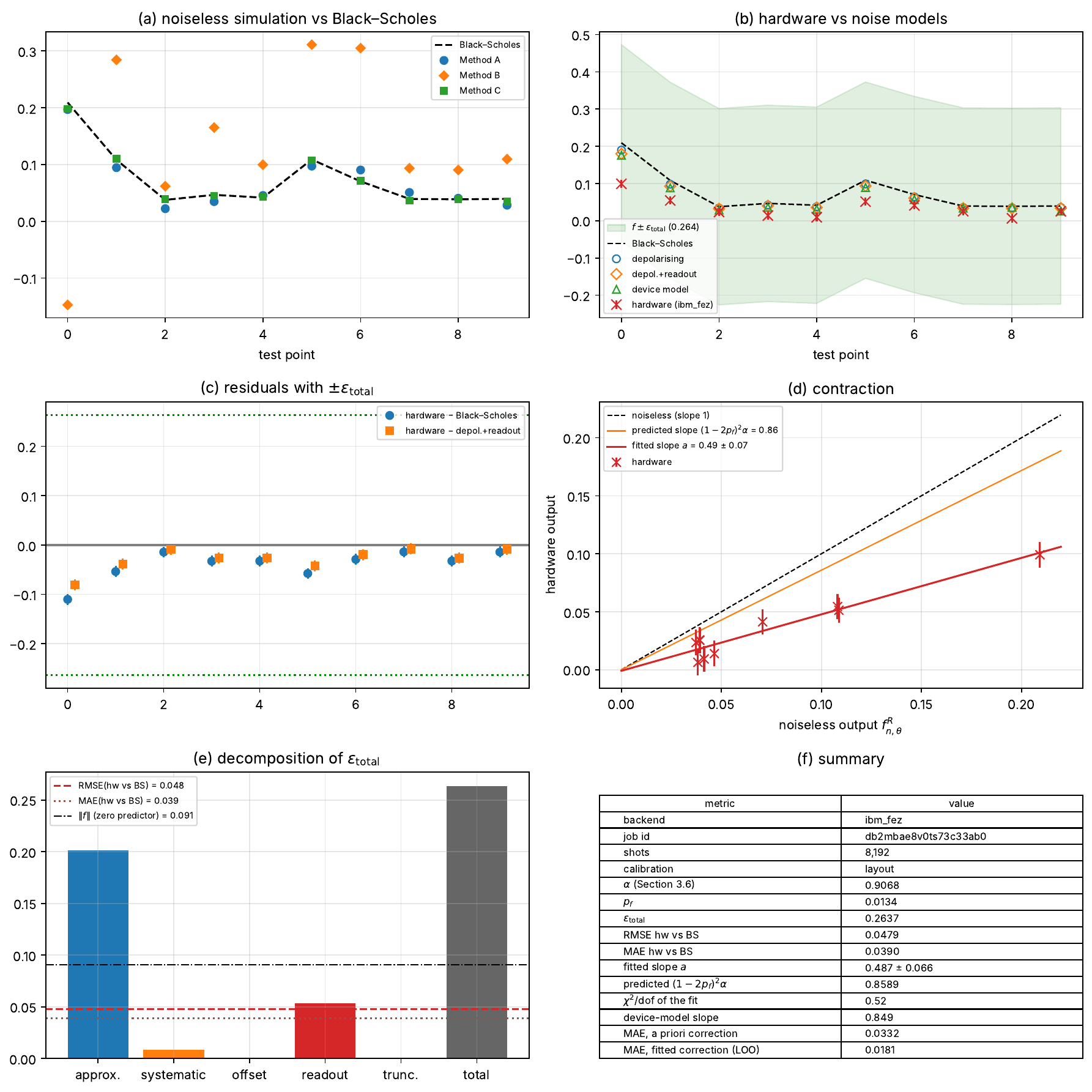}
\caption{Hardware execution on \texttt{ibm\_fez} (job \texttt{db2mbae8v0ts73c33ab0}; $n=8$, $8192$ shots, $10$ test points).
\emph{Panel~(a)}: noiseless simulation ($8192$ shots) of Methods~A, B and~C against the Black-Scholes price.
\emph{Panel~(b)}: hardware output (crosses, with the shot-noise error bars~\eqref{eq:stat_error}) against the predictions of the depolarising model (Corollary~\ref{cor:depolar_qnn_output}), of the depolarising model with readout (Corollary~\ref{cor:readout_exact}) and of the \texttt{qiskit-aer} noise model built from the same calibration; the band is $f_2\pm\eps_{\mathrm{total}}$, with $\eps_{\mathrm{total}}$ from~\eqref{eq:eps_total}.
\emph{Panel~(c)}: residuals with respect to the Black-Scholes price and to the depolarising model with readout, with $\pm\eps_{\mathrm{total}}$ (dotted).
\emph{Panel~(d)}: hardware output against the noiseless output $f^R_{n,\theta}$, with the noiseless diagonal, the predicted slope $(1-2\pf)^2\alpha=0.86$ and the fitted slope $a=0.49\pm0.07$.
\emph{Panel~(e)}: decomposition of $\eps_{\mathrm{total}}$ into approximation, systematic, offset, readout and truncation terms, with the RMSE and MAE of the hardware output with respect to the Black-Scholes price and the error $\|f_2\|_{L^2(\mu)}$ of the trivial predictor~$0$.
\emph{Panel~(f)}: summary, including the a priori ($\beta_1=1/\bar\alpha$) and fitted (leave-one-out) corrections.}
\label{fig:hardware_main}
\end{figure}

\subsubsection*{The bound}
For the Black-Scholes Put, Proposition~\ref{prop:full_noise_approximation} and the estimate~\eqref{eq:bs_bound} give
\begin{equation}\label{eq:eps_total}
\eps_{\mathrm{total}} = \frac{\alpha \Bf(a^*)}{\sqrt{n}} + (1-\alpha)\|f_2\|_{L^2(\mu)}  + R(1-\alpha)\left(1 - \frac{4n}{2^{\nf}}\right)  + 4R\pf+\delta_{\mathrm{trunc}},
\end{equation}
with $\mu$ the empirical measure of the $10$ test points, for which $\|f_2\|_{L^2(\mu)}=0.0906$ (against $0.103$ on the $100$ test points of Section~\ref{subsec:bs_expt} and $0.092$ on the $20$ points of Section~\ref{subsec:noise_expt}, which moves the threshold of Remark~\ref{rem:three_terms} from $38$ to $48$). The five terms are $0.2014$ (approximation), $0.0085$ (systematic), $0$ (offset, since $n_0=0$), $0.0535$ (readout) and $0.0004$ (truncation), so that $\eps_{\mathrm{total}}=0.2637$, or $0.2640$ with shot noise (Remark~\ref{rem:shots}); the sharper bound~\eqref{eq:readout_sharp} gives $0.2040$. The hardware output deviates from the Black-Scholes price by $0.048$ in root mean square ($0.039$ in mean absolute value, $0.110$ at most), so the bound holds (Figure~\ref{fig:hardware_main}(b)--(c)). As in Section~\ref{subsec:bs_expt}, however, it is not informative at $n=8$: it exceeds $\|f_2\|_{L^2(\mu)}$, the error of the trivial predictor~$0$, so this comparison is a weak test. Executing the informative regime $n=32$ ($\nf=7$) on \texttt{ibm\_fez} would require $2^{\nf}-4=124$ CNOTs before routing ($N_{\Ug}=190$ gates in Section~\ref{subsec:noise_expt}), which at the observed per-gate effective error rate ($\sim 1.7\%$) would reduce the contraction factor~$\bar\alpha$ below the threshold $\bar\alpha R\geqslant L^1[\widehat f_1]=0.42$ of Corollary~\ref{cor:noise_aware}; we therefore use $n=8$ on hardware, where the structure of the noise is directly informative.

\subsubsection*{Structure of the noise}
By Corollary~\ref{cor:readout_exact} with $n_0=0$, the depolarising model with readout predicts that the hardware output is the contraction $\bar\alpha f^R_{n,\theta}$ of the noiseless output, with $\bar\alpha=(1-2\pf)^2\alpha=0.859$. A weighted least-squares fit of the hardware output against $af^R_{n,\theta}+b$, with weights given by~\eqref{eq:stat_error}, yields $a=0.487\pm0.066$ and $b=-0.0009\pm0.0060$, with $\chi^2/\mathrm{dof}=0.52$. The data are thus consistent with a pure contraction, with zero intercept as predicted for $n_0=0$, and the residuals around the fitted line are consistent with shot noise. The contraction factor, however, lies $5.6$ standard errors below the prediction. This conclusion does not hinge on the fit: six of the ten test points have prices between $0.037$ and $0.047$, so the slope relies on the point $S=0.80$, but without it $a=0.52\pm0.13$, still $2.6$ standard errors below the prediction, and a fit through the origin, as the model prescribes, gives $a=0.478\pm0.038$. Test points spread over the price range would sharpen this estimate. Once the readout factor is removed, it corresponds to $\alpha=0.51\pm0.07$, that is, to an effective depolarisation $1-\alpha\approx0.49$ instead of the $0.093$ computed from calibration data; attributed entirely to the $39$ two-qubit gates, it would correspond to an error of $1.7\%$ per gate, about eight times the calibrated value.
This gap is not due to the simplicity of the recipe of Section~\ref{subsec:hardware_params}: a gate-level simulation with the noise model that \texttt{qiskit-aer} builds from the same calibration data (thermal relaxation, depolarising and readout errors on every physical gate of the transpiled circuits) gives a contraction factor, the least-squares slope against~$f^R_{n,\theta}$, of $0.85$, and counting every physical one-qubit gate and decoherence on all five qubits only lowers the prediction to $0.78$. The calibration data therefore underestimate the noise of this circuit. We cannot identify the origin from the present data; candidates include coherent and crosstalk errors that randomised benchmarking of isolated gates does not capture, errors on idle qubits during the routed circuit, which ran without dynamical decoupling, and drift since the last calibration. We note that the Pearson correlation between hardware and model outputs ($0.96$ for both models) is invariant under affine maps and therefore cannot detect such a discrepancy.

\subsubsection*{Consequences for the bound}
Since the calibrated~$\alpha$ is optimistic, the validity of the bound above requires an argument. By Remark~\ref{rem:three_terms}, $\eps_{\mathrm{total}}$ increases with~$\alpha$ as long as $n<(\Bf(a^*)/\|f_2\|_{L^2(\mu)})^2\approx48$, and $n=8$ lies in this regime: with the measured $\alpha=0.51$, \eqref{eq:eps_total} would give $0.212$ instead of $0.264$. The calibrated bound thus remains valid, but only because~$n$ is small; for larger~$n$, where $\eps_{\mathrm{total}}$ decreases with~$\alpha$, the same miscalibration could make $\eps_{\mathrm{total}}$ too small, and~$\alpha$ should then be estimated from data. For noise-aware parameters the question does not arise: by Corollary~\ref{cor:noise_aware}, parameters attaining the noiseless bound exist as long as $\bar\alpha R\geqslant L^1[\widehat f_1]=0.42$, which holds for the measured $\bar\alpha\approx0.49$ at $R=1$, whatever its precise value. With the measured contraction, on the other hand, the structure of the bound accounts for the observed error: the noiseless output is within $4\cdot10^{-4}$ of the price, so the error is the systematic term of~\eqref{eq:readout_sharp} at the measured contraction, $(1-a)\|f_2\|_{L^2(\mu)}=0.047$, combined with shot noise of standard deviation $0.011$, and $\sqrt{0.047^2+0.011^2}=0.048$ is the observed value.

\subsubsection*{Noise cancellation}
Theorem~\ref{cor:noise_correction} and Corollary~\ref{cor:readout_exact} suggest the correction $\beta_1=1/\bar\alpha$, $\beta_2=0$. With the calibrated $\bar\alpha=0.859$ ($\beta_1=1.16$), this a priori correction reduces the MAE from $0.039$ to $0.033$ only. Fitting $(\beta_1,\beta_2)$ by least squares to the Black-Scholes prices of the other nine test points (leave-one-out cross-validation, since this job did not include training points) reduces it to $0.018$ (RMSE $0.021$); on all ten points, the fit gives $\beta_1=1.90$, i.e.\ $1/\beta_1=0.53$, consistent with the measured contraction. The remaining error is shot noise amplified by~$\beta_1$ (Remark~\ref{rem:shot_amplification}): $1.90\times0.011=0.021$. On current hardware, the two parameters of Theorem~\ref{cor:noise_correction} should therefore be fitted to hardware outputs at points with known targets rather than computed from calibration data; the notebook now executes training points in the same job and fits the correction on them, as recommended in Remark~\ref{rem:practical}. This comes at the cost of increasing the number of shots by at least a factor of~$\beta_1^2$ to retain the statistical accuracy of the noiseless circuit (Remark~\ref{rem:shot_amplification}).

\section{Conclusion and outlook}\label{sec:conclusion}

We have derived a quantitative universal approximation theorem with explicit, non-asymptotic error bounds for noisy quantum neural networks. The bound decomposes the error into three transparent contributions -- a noiseless approximation term scaled by the hardware fidelity $\alpha$, a systematic term $(1-\alpha)\|f\|_{L^2(\mu)}$, and an offset depending on the qubit count -- to which a~$4R\pf$ readout term is added when classical measurement error is included; readout error in fact acts as a further contraction of the output (Corollary~\ref{cor:readout_exact}). Specialising to depolarising noise calibrated to real hardware, we showed that a two-parameter affine post-processing layer (Theorem~\ref{cor:noise_correction}) exactly cancels the bias, and that parameters trained against the noisy output attain the noiseless bound (Corollary~\ref{cor:noise_aware}), giving concrete prescriptions for noise-aware training without modifying the underlying circuit.

We focused on functions given as expectations $f(\xx)=\EE[\Phi(\xx+L)]$, motivated by quantitative finance. For exponential L\'evy models, we obtained explicit constants in terms of the characteristic function of $L_T$, and worked out explicit bounds for Black-Scholes Put pricing. On the IBM Heron r2 processor \texttt{ibm\_fez}, the hardware output has the affine form predicted by the depolarising model, with no detectable offset, and its deviation from the Black-Scholes price lies inside the analytical envelope. Its contraction factor, $0.49\pm0.07$, is however much smaller than the value $0.86$ predicted from calibration data, both by our formulae and by a gate-level noise model built from the same data: calibration data underestimate the effective noise of the circuit. The bound remains valid because, for $n=8$, it increases with~$\alpha$; this would no longer hold for $n$ above about~$48$, where the bound decreases with~$\alpha$. Fitting the two parameters of the affine correction to data halves the error, which is then dominated by shot noise. Alternatively, parameters trained against the noisy output remove the systematic term (Corollary~\ref{cor:noise_aware}); in simulation at $n=32$, where the bound is informative, they keep the noiseless accuracy for every $\alpha\geqslant0.42$.

Several questions remain open. Estimating~$\alpha$ from data rather than from calibration, for instance with noise-estimation circuits of the same structure as the QNN~\cite{urbanek2021mitigating,vovrosh2021simple}, would make the bound reliable in the informative regime, and Pauli twirling~\cite{wallman2016noise} and dynamical decoupling~\cite{viola1999dynamical} would bring the hardware noise closer to the depolarising model and help identify the origin of the discrepancy. 
The depolarising specialisation could be extended to non-unital noise channels (amplitude damping in isolation, leakage), where the recent results of~\cite{mele2024noise,franca2025efficient} suggest qualitatively different behaviour. Finally, the architecture of~\cite{gonon2023universal} is one of several Fourier-type ans\"atze; adapting our analysis to data-reuploading circuits~\cite{PerezSalinas2020datareuploading,Schuld2021} and to fault-tolerant regimes are natural next steps.
\newpage
\appendix
\section{Noise parameter table}\label{app:table}
The following table summarises the error parameters explained in Section~\ref{subsec:hardware_params}.

\begin{table}[h]
\centering\small    \renewcommand{\arraystretch}{1.3}
\resizebox{\textwidth}{!}{%
\begin{tabular}{|c|c|c|c|c|}
\hline
\textbf{Hardware} & \textbf{Parameter} & \textbf{Symbol} & \textbf{Value} & \textbf{Source} \\
\hline
\multirow{14}{*}{\textbf{IBM \texttt{ibm\_fez}}}
    & Single-qubit gate error ($\sqrt{\mathtt{X}}$) & $\eps_{1\Qg}$ & $2.30\cdot 10^{-4}$ & \cite{ibm_quantum_platform} \\
    & Two-qubit gate error ($\mathtt{CZ}$) & $\eps_{2\Qg}$ & $2.01\cdot 10^{-3}$ & \cite{ibm_quantum_platform} \\
    & Readout error (target qubits) & $\pf$ & $0.0134$ & \cite{ibm_quantum_platform} \\
    & Relaxation time & $T_1$ & $150.3\,\mu$s & \cite{ibm_quantum_platform} \\
    & Dephasing time & $T_2$ & $131.9\,\mu$s & \cite{ibm_quantum_platform} \\
    & Two-qubit gate duration & $t_{2\Qg}$ & $68$ ns & \cite{ibm_quantum_platform} \\
    & Two-qubit gates (transpiled circuit) & $N_{2\Qg}$ & $39$ & Transpiler \\
    & Circuit execution time & $t_{\text{circ}}$ & $2.652\,\mu$s & Computed \\
    & Relaxation error probability & $p_{T_1}$ & $0.0175$ & Computed \\
    & Dephasing error probability & $p_{T_2}$ & $0.0199$ & Computed \\
    & Gate contribution to $\lambda_{\Ug}$ & $\lambda_{\Ug}^{\mathrm{gate}}$ & $0.0754$ & Computed \\
    & Effective $\Vg$ parameter & $\lambda_{\Vg}$ & $6.9\cdot 10^{-4}$ & Computed \\
    & Effective $\Ug$ parameter & $\lambda_{\Ug}$ & $0.0926$ & Computed \\
    & Fidelity factor & $\alpha$ & $0.9068$ & Computed \\
    \hline
    \multirow{6}{*}{\textbf{Quantinuum H2}} 
    & Single-qubit gate infidelity & $\eps_{1\Qg}$ & $3 \cdot  10^{-5}$ & \cite{Quantinuum_H2_2025} \\
    & Two-qubit gate infidelity & $\eps_{2\Qg}$ & $1 \cdot  10^{-3}$ & \cite{Quantinuum_H2_2025} \\
    & Single-qubit gate duration & $t_{1\Qg}$ & $\sim 5\, \mu$s & \cite{Quantinuum_H2_2025} \\
    & Two-qubit gate duration & $t_{2\Qg}$ & $\sim 100\, \mu$s & \cite{Quantinuum_H2_2025} \\
    & Memory error per qubit (depth-1) & $\lambda_{\text{mem}}$ & $2 \cdot  10^{-4}$ & \cite{Quantinuum_H2_2025} \\
    & State prep/meas error & $r$ & $1 \cdot  10^{-3}$ & \cite{Quantinuum_H2_2025} \\
    \hline
    \multirow{4}{*}{\textbf{Rigetti Cepheus}} 
    & Two-qubit gate error & $\eps_{2\Qg}$ & $5 \cdot  10^{-3}$ & \cite{Rigetti2025_36Q} \\
    & Gate speed & $t_{2\Qg}$ & $60$--$80$ ns & \cite{Rigetti_Q2_2025} \\
    & $T_1$ (typical superconducting) & $T_1$ & $\sim 50$--$100\, \mu$s & \cite{Rigetti2025_36Q} \\
    & $T_2$ (typical superconducting) & $T_2$ & $\sim 20$--$50\, \mu$s & \cite{Rigetti2025_36Q} \\
    \hline
    \end{tabular}}
    \caption{Hardware parameters and the depolarising parameters of Section~\ref{subsec:hardware_params}. IBM \texttt{ibm\_fez} (Heron r2, 156 qubits): calibration of the physical qubits and couplers used by the transpiled circuit, as reported when the job was submitted (job \texttt{db2mbae8v0ts73c33ab0}); gate errors, $T_1$ and $T_2$ are averages over these qubits and couplers, and~$\pf$ over the two target qubits. Quantinuum H2 (56 qubits, trapped ion) and Rigetti Cepheus-1-36Q (36 qubits, modular superconducting): published specifications, reported for comparison only; the Rigetti $T_1$ and $T_2$ are indicative ranges for superconducting qubits, not device-specific values.}
\label{tab:noise_params}
    \end{table}

\clearpage 

\section{Additional figures}\label{app:figures}
\subsection{Gaussian experiments}
The following graphs illustrate the numeric of Section~\ref{subsec:gaussian_expt}.

\begin{figure}[H]
\centering
\includegraphics[width=0.7\linewidth]{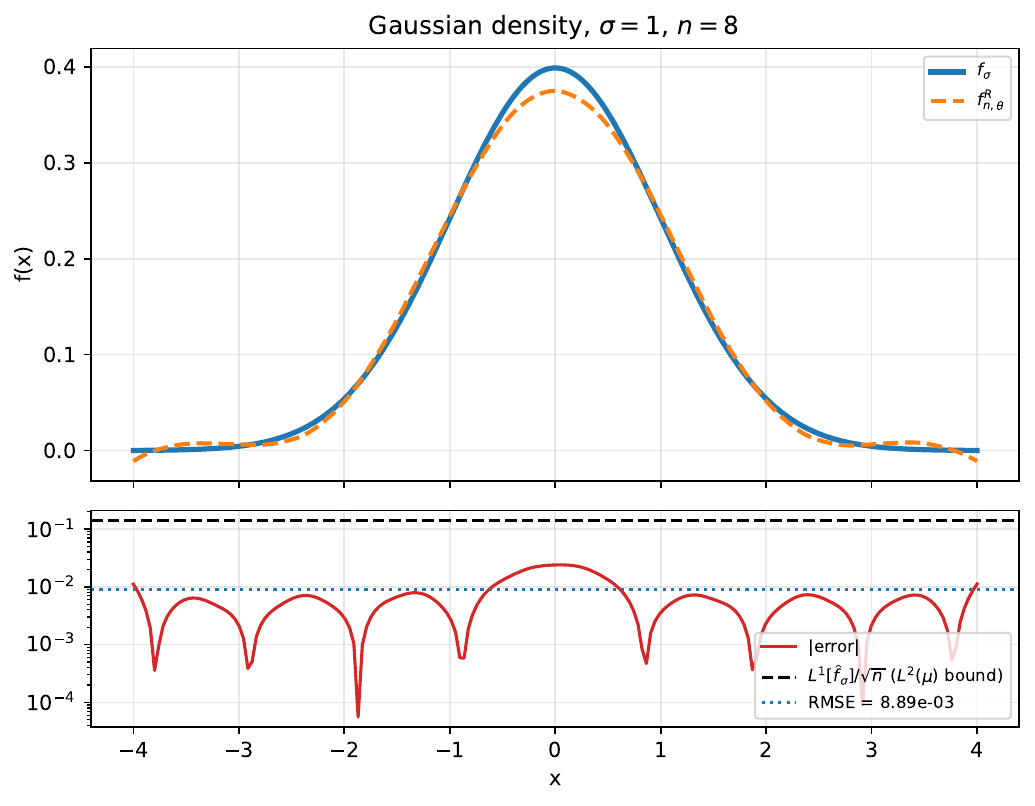}
\caption{Gaussian density approximation (Statement~\ref{st:state}), part~1 of~2: $(\sigma,n)=(1,8)$, $R=0.6$. Top: true density and QNN output; bottom: pointwise $|\mathrm{error}|$ (log scale), with the $L^2(\mu)$ bound $L^1[\widehat{f}_\sigma]/\sqrt{n}$ (dashed) and the RMSE (dotted).}
\label{fig:gaussian_part1}
\end{figure}

\begin{figure}[H]
\centering
\begin{subfigure}{\linewidth}
\centering
\includegraphics[width=\linewidth]{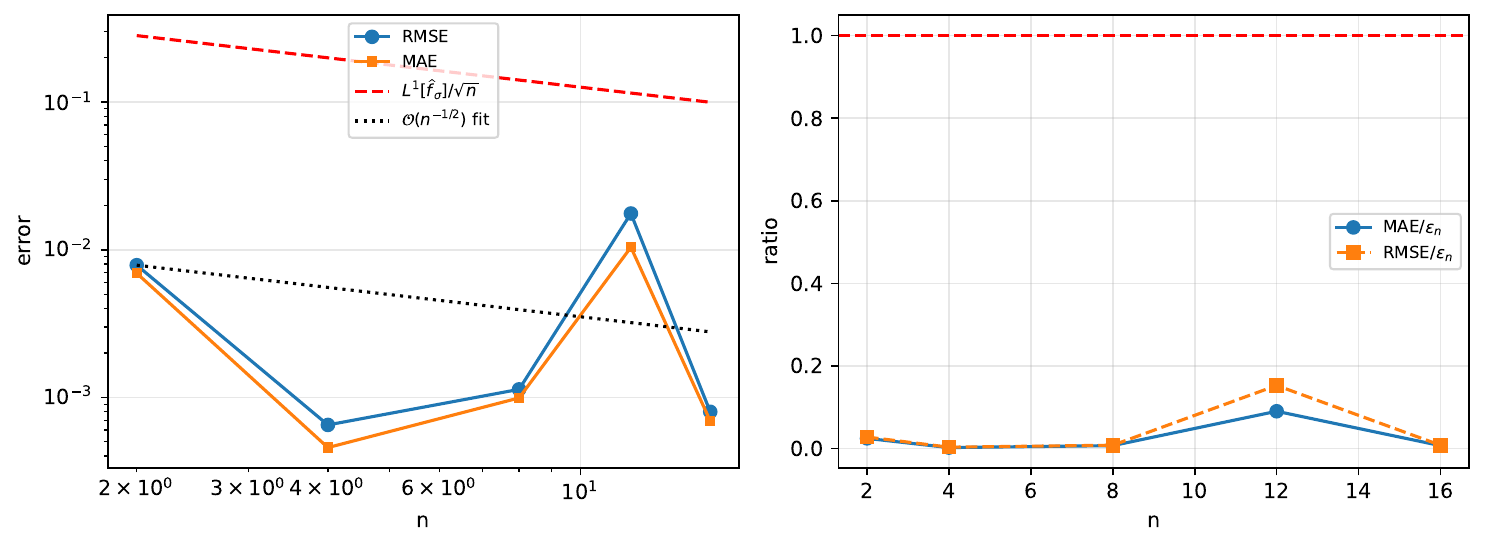}
\caption{Left: RMSE and MAE against~$n$ (log-log), with the bound $L^1[\widehat{f}_\sigma]/\sqrt{n}$ (red dashed) and an $\Oo(n^{-1/2})$ reference (dotted). Right: ratios MAE$/\eps_n$ and RMSE$/\eps_n$.}
\label{fig:gaussian_convergence}
\end{subfigure}
\vspace{0.8cm}
\begin{subfigure}{\linewidth}
\centering
\includegraphics[width=0.7\linewidth]{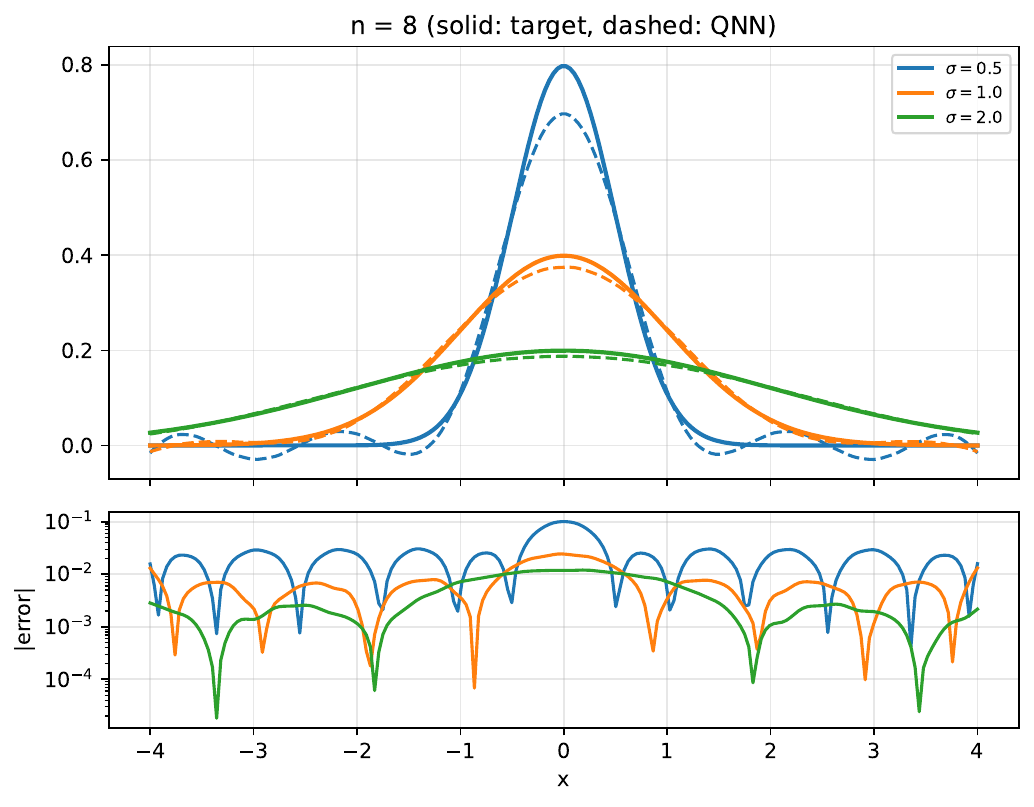}
\caption{$\sigma \in \{0.5, 1, 2\}$, $n = 8$. Top: target (solid),
QNN output (dashed); bottom: $|\mathrm{error}|$ (log scale).}
\label{fig:multi_sigma}
\end{subfigure}
\caption{Gaussian density approximation (Statement~\ref{st:state}), part~2 of~2.}
\label{fig:gaussian_part2}
\end{figure}

\clearpage

\subsection{Black-Scholes Put option in the noiseless regime}
Numerical illustrations for the noiseless Black-Scholes Put of Section~\ref{subsec:bs_expt}.

\begin{figure}[H]
    \centering

    \begin{subfigure}{\linewidth}
        \centering
        \includegraphics[width=0.85\linewidth]{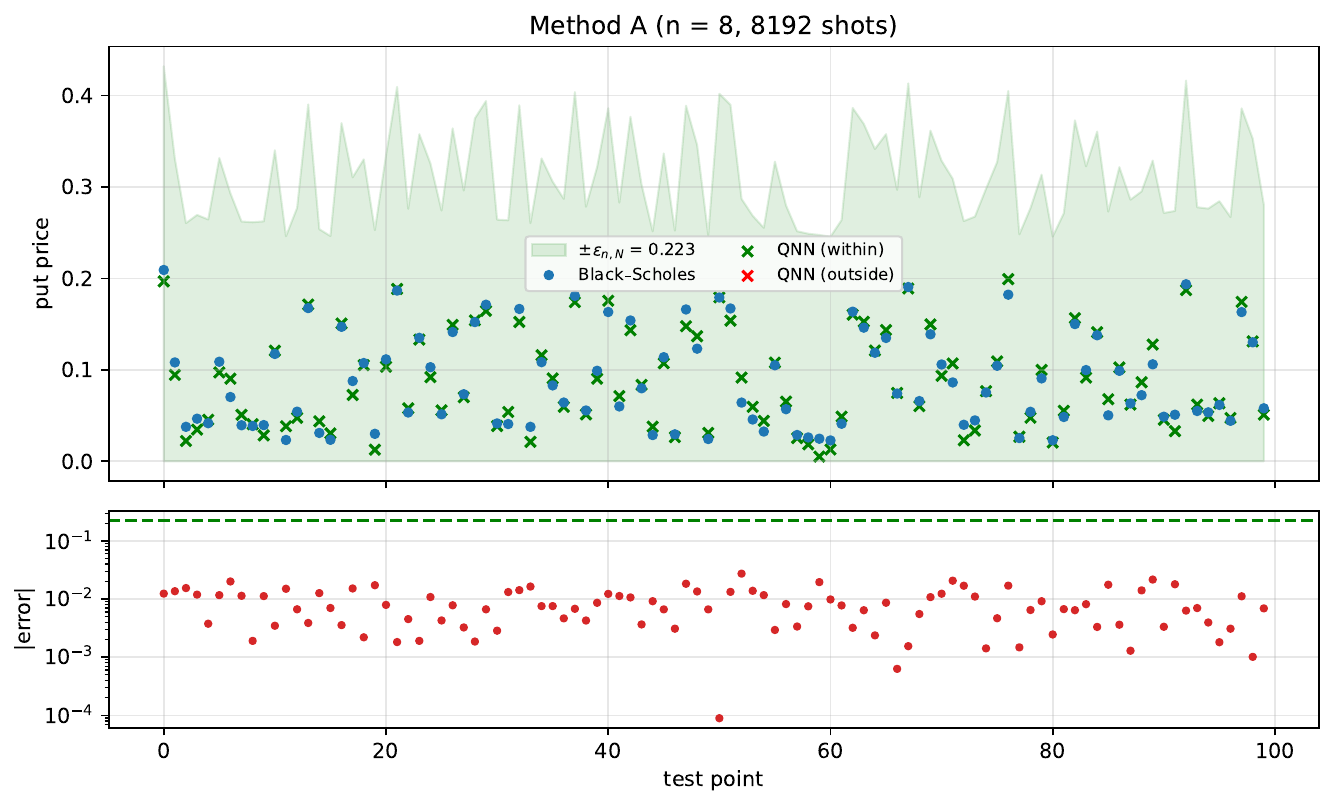}
        \caption{Method~A: QNN outputs ($8192$ shots) and Black-Scholes prices on the $100$ test points (top), and absolute errors (bottom, log scale). Green: within the band; red: outside.}
        \label{fig:method_A}
    \end{subfigure}

    \vspace{0.8cm}
    \begin{subfigure}{\linewidth}
        \centering
        \includegraphics[width=0.85\linewidth]{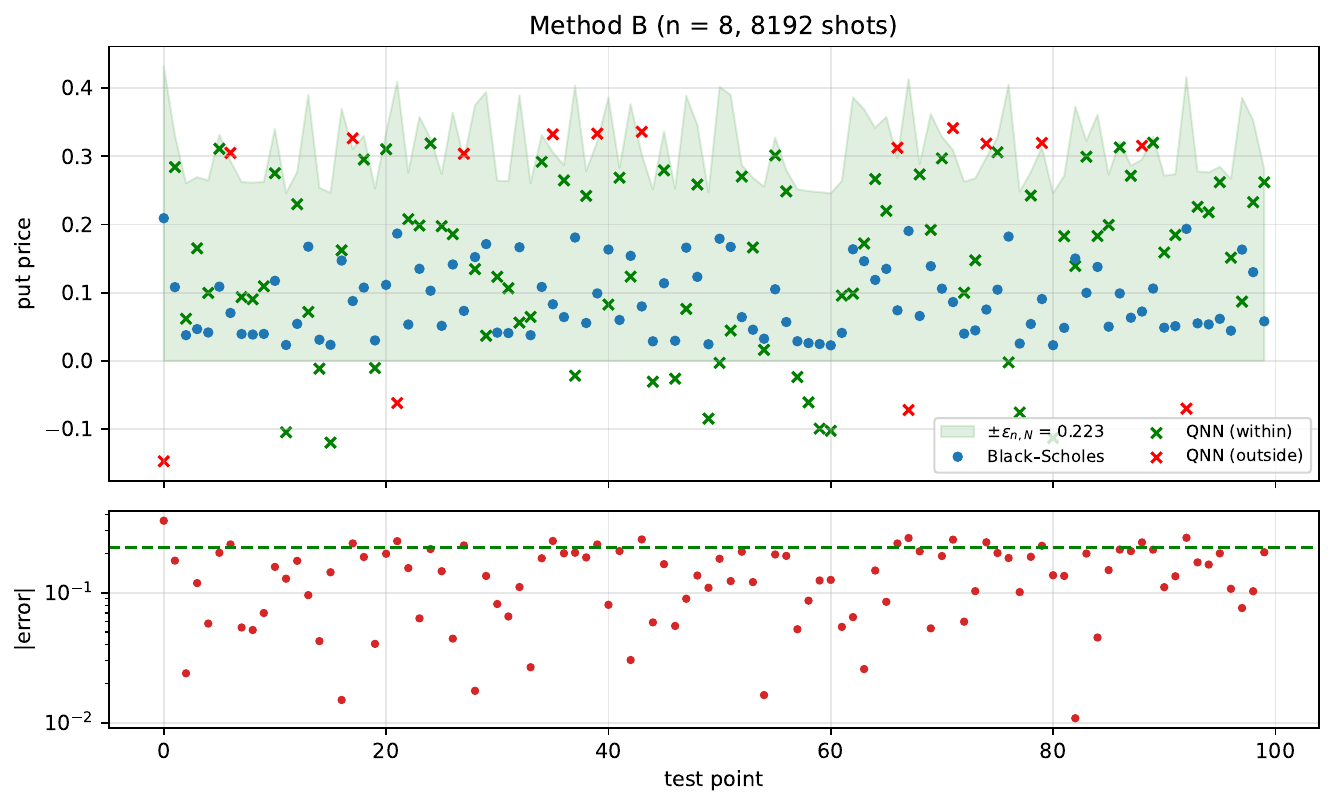}
        \caption{Method~B: same layout as~(a). Optimising the last block alone is less accurate than the trivial predictor~$0$.}
        \label{fig:method_B}
    \end{subfigure}

    \caption{Per-method noiseless performance. 
    The band is $\pm\eps_{n,N}$ of Remark~\ref{rem:shots}.}
    \label{fig:method_compare_part1}
\end{figure}

\begin{figure}[H]
\centering
\includegraphics[width=\linewidth]{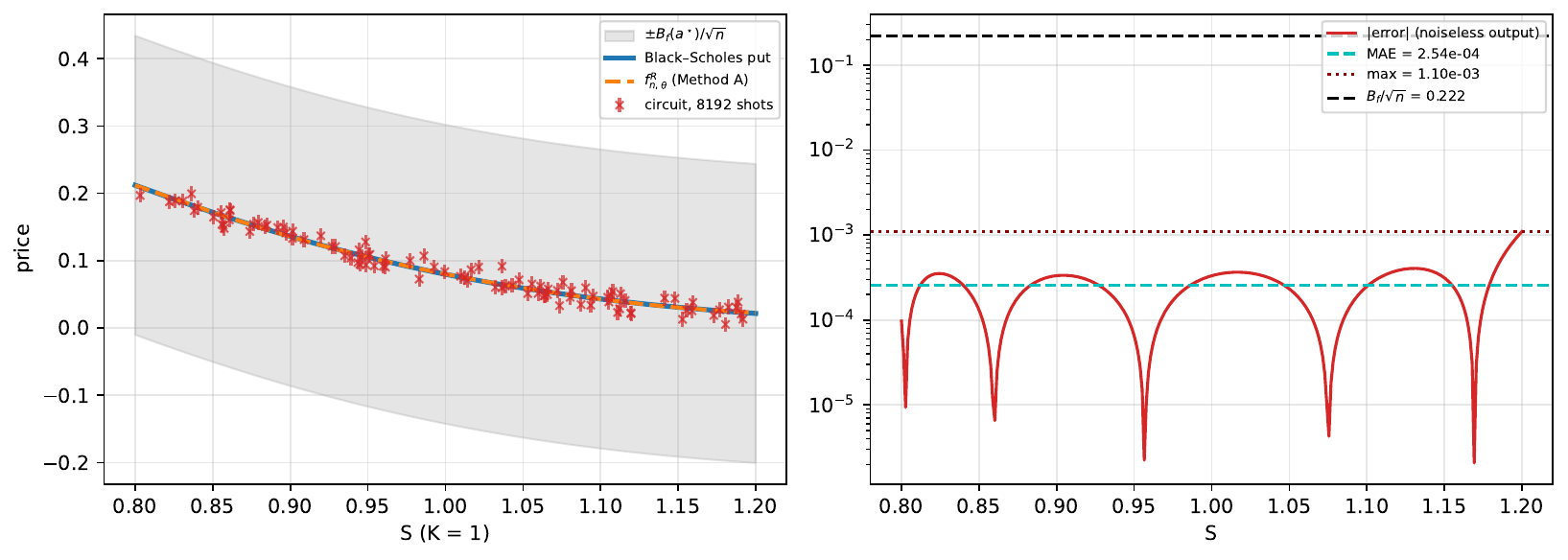}
\caption{Black-Scholes Put ($K=1$, $\sigma=0.2$, $T=1$, $r=0$) and QNN output (Method~A, $n=8$, $\nf=5$ qubits). Left: price, closed-form QNN output, envelope $\pm\Bf(a^*)/\sqrt n$, and circuit outputs with $8192$ shots; right: pointwise error of the closed-form output, with its mean, its maximum and the bound $\Bf(a^*)/\sqrt n$.}
\label{fig:bs_curve}
\end{figure}

\begin{figure}[H]
\includegraphics[width=\linewidth]{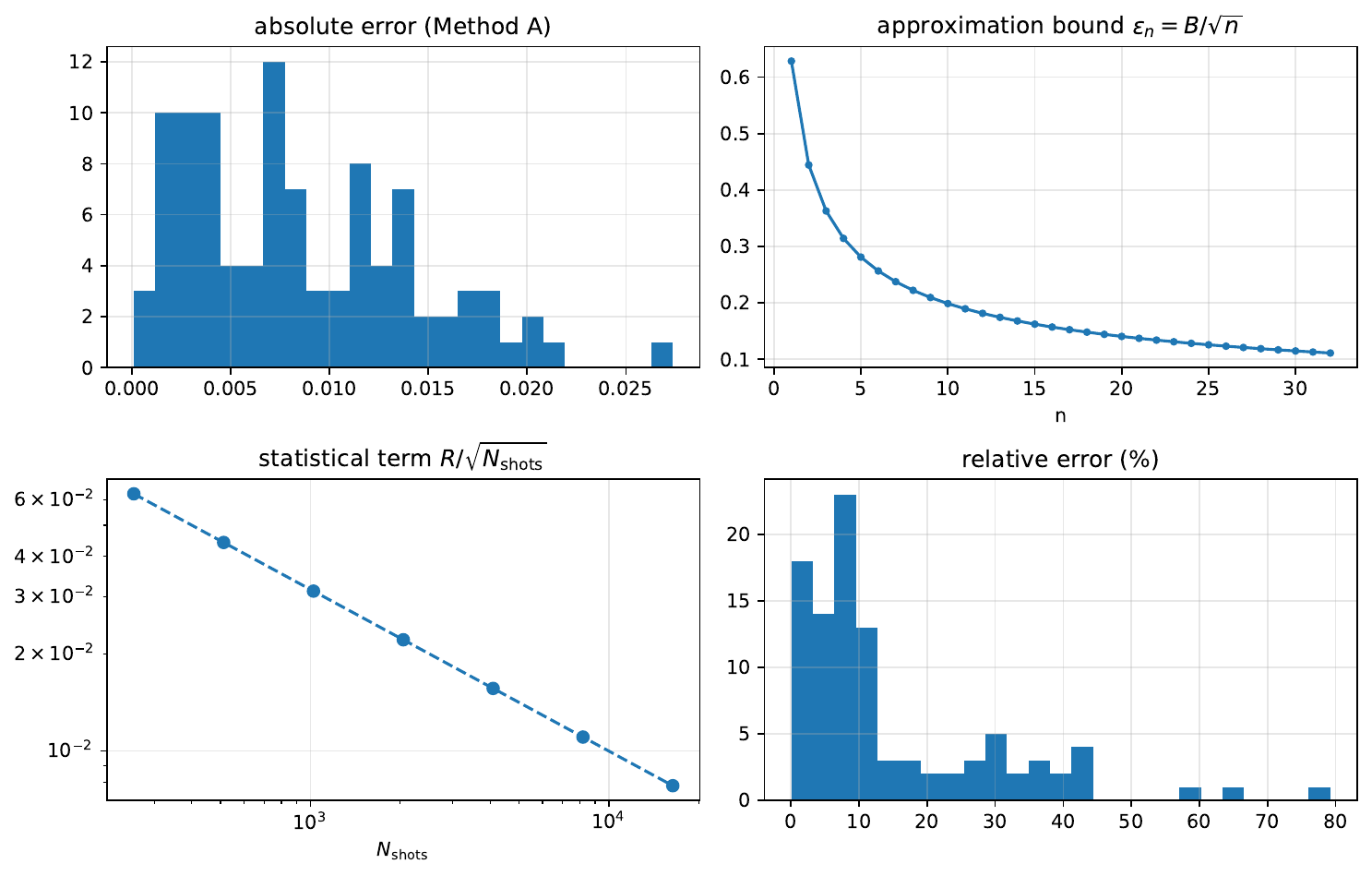}
        \caption{Error analysis for Method~A: absolute errors with $8192$ shots; approximation bound $\eps_n$ against~$n$; statistical term $R/\sqrt{N_{\mathrm{shots}}}$ against~$N_{\mathrm{shots}}$; relative errors.}
        \label{fig:error_A}
\end{figure}

\begin{figure}[H]
  \centering
  \includegraphics[width=\linewidth]{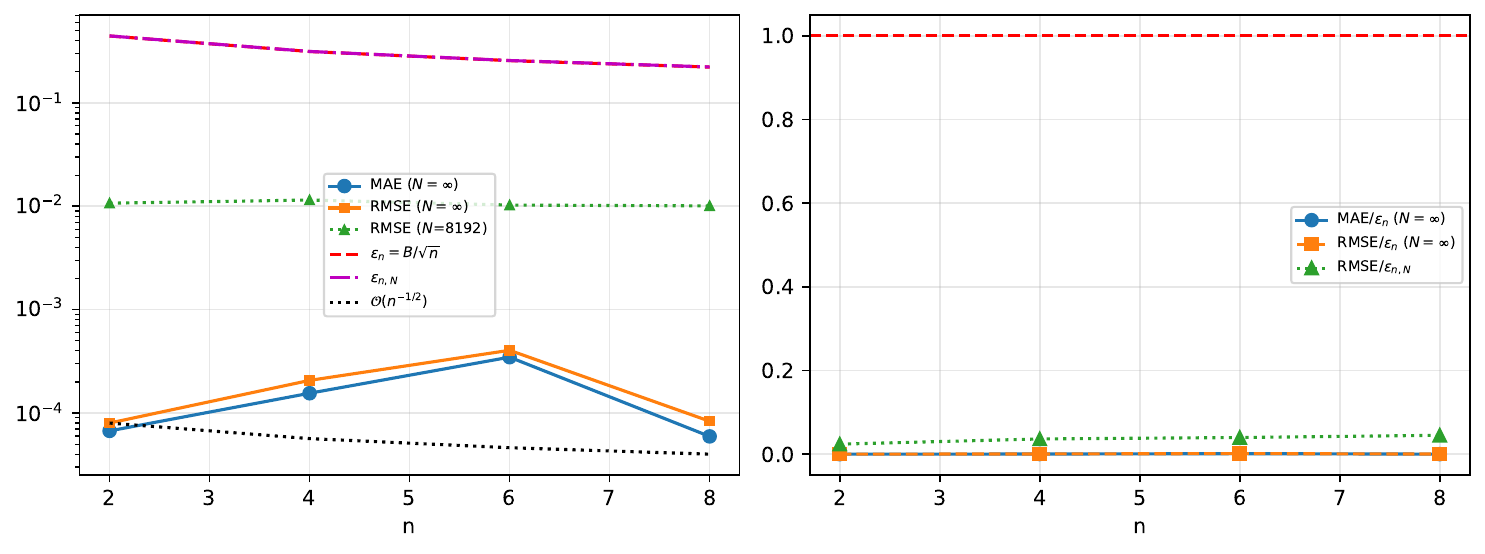}
  \caption{Convergence in~$n$ (Method~A, noiseless, $n\in\{2,4,6,8\}$). Left: MAE and RMSE of the closed-form output ($N=\infty$) and RMSE with $8192$ shots, against $\eps_n=\Bf(a^*)/\sqrt n$ (red dashed) and $\eps_{n,N}$ (dash-dotted), with an $\Oo(n^{-1/2})$ reference (dotted). Right: ratios MAE$/\eps_n$ and RMSE$/\eps_n$ ($N=\infty$), and RMSE$/\eps_{n,N}$.}
  \label{fig:bs_scaling}
\end{figure}

\subsection{Figures for the noise simulation}\leavevmode
\begin{figure}[H]
  \centering
  \begin{subfigure}{0.9\textwidth}
    \centering
    \includegraphics[scale=0.35]{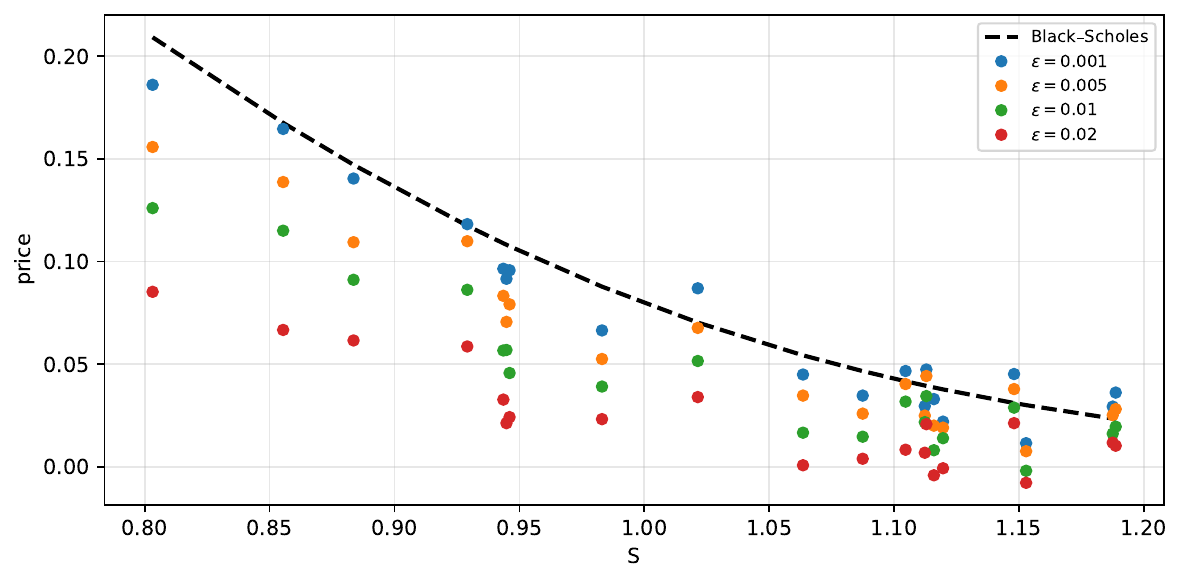}
    \caption{Noisy QNN outputs ($8192$ shots) against Black-Scholes Put (black, dashed) on the $20$ test points. Larger~$\eps$ contracts the output towards~$0$ (Corollary~\ref{cor:depolar_qnn_output} with $n_0=0$).}
    \label{fig:noise_curves}
  \end{subfigure}

  \medskip

  \begin{subfigure}{0.9\textwidth}
    \centering
    \includegraphics[scale=0.35]{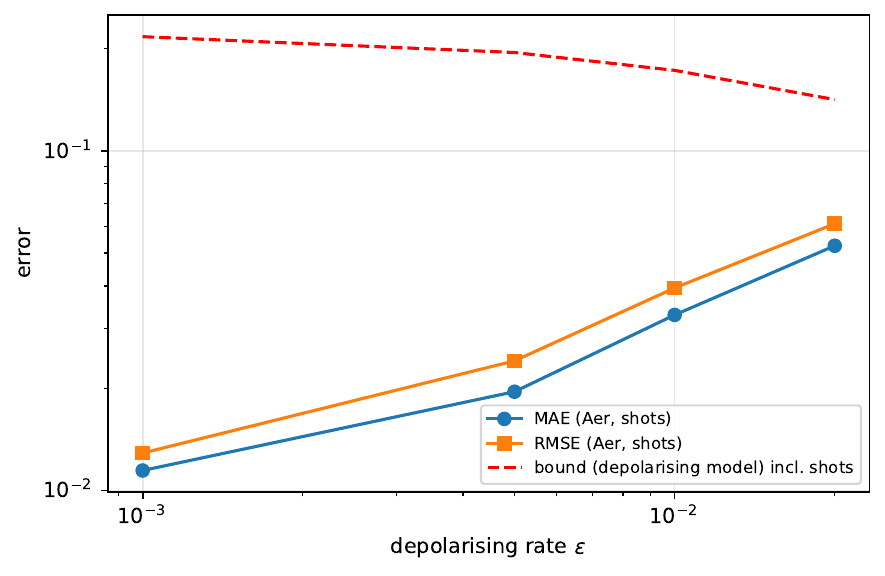}
    \caption{MAE and RMSE against the depolarising rate~$\eps$ (log-log), with the bound of Theorem~\ref{thm:depolar_approximation} including shot noise (dashed).}
    \label{fig:noise_rmse}
  \end{subfigure}
  \caption{Gate-level depolarising noise simulation of Section~\ref{subsec:noise_expt} ($n=8$, Method~A,
    $N_{\mathrm{shots}} = 8192$, $20$ test points).}
  \label{fig:noise_sim}
\end{figure}

\bibliographystyle{siam}
\bibliography{references}
\end{document}